\documentclass[12pt,a4paper]{article}

\usepackage{authblk}
\usepackage{graphicx}
\usepackage{url}
\usepackage{amstext}
\usepackage{amsmath}
\usepackage{amsfonts}
\usepackage{subcaption}
\usepackage{mathtools}
\usepackage{float}
\usepackage{fancyhdr}
\usepackage{amsthm}
\usepackage{multirow}
\usepackage{cite}
\usepackage{placeins}
\usepackage{makecell}

\usepackage{tikz}

\newtheorem{theorem}{Theorem}[section]
\newtheorem{definition}[theorem]{Definition}
\newtheorem{claim}[theorem]{Claim}
\newtheorem{proposition}[theorem]{Proposition}

\newtheorem{lemma}[theorem]{Lemma}

\usepackage[margin=1.0in]{geometry}
\usepackage{caption}
\title{\bf If time had no beginning: growth dynamics for past-infinite causal sets}
\author[1]{Bruno Valeixo Bento}
\author[2,3]{Fay Dowker}
\author[2]{Stav Zalel}
\affil[1]{\small Department of Mathematical Sciences, University of Liverpool, Liverpool L69 7ZL}
\affil[2]{Blackett Laboratory, Imperial College London, SW7 2AZ, U.K.}
\affil[3]{Perimeter Institute, 31 Caroline Street North, Waterloo ON, N2L 2Y5, Canada.}
\usepackage[hidelinks]{hyperref}
\newcommand{\lc}[1]{\tilde{#1}}
\newcommand{\Z}{\mathbb{Z}}
\newcommand{\N}{\mathbb{N}}

\usepackage{soul}

\def\twoch{\,\,\begin{picture}(0,1) 
\thicklines
\multiput(0,0)(0,10){2}{\circle*{2}}
\put(0,0){\line(0,1){10}}
\end{picture}\,\,}

\def\threech{\,\,\begin{picture}(0,1) 
\thicklines
\multiput(0,0)(0,10){3}{\circle*{2}}
\put(0,0){\line(0,1){10}}
\put(0,10){\line(0,1){10}}
\end{picture}\,\,}

\def\twoach{\,\begin{picture}(1,1) 
\thicklines
\multiput(0,0)(5,0){2}{\circle*{2}}
\end{picture}\,\,\,}

\def\threeach{\,\begin{picture}(1,1) 
\thicklines
\multiput(0,0)(10,0){3}{\circle*{2}}
\end{picture}\,\,\,}

\def\oneach{\,\begin{picture}(1,1) 
\thicklines
\multiput(0,0)(10,0){1}{\circle*{2}}
\end{picture}\,\,\,}

\def\lambdacauset{\,\begin{picture}(1,1) 
\thicklines
\multiput(0,0)(6.5,0){2}{\circle*{2}}
\multiput(3.25,10.2)(3.25,10.2){1}{\circle*{2}}
\put(0,0){\line(1,3){3.5}}
\put(6.5,0){\line(-1,3){3.5}}
\end{picture}\,\,\,}

\def\lambdafour{\,\begin{picture}(1,1) 
\thicklines
\multiput(0,0)(6.5,0){2}{\circle*{2}}
\multiput(3.25,10.2)(3.25,10.2){1}{\circle*{2}}
\multiput(3.25,20)(3.25,20){1}{\circle*{2}}
\put(0,0){\line(1,3){3.5}}
\put(6.5,0){\line(-1,3){3.5}}
\put(3.25,10.2){\line(0,1){10}}
\end{picture}\,\,\,}

\def\fourch{\,\,\begin{picture}(0,1) 
\thicklines
\multiput(0,0)(0,10){4}{\circle*{2}}
\put(0,0){\line(0,1){10}}
\put(0,10){\line(0,1){10}}
\put(0,20){\line(0,1){10}}
\end{picture}\,\,}

\def\vee{\,\,\begin{picture}(0,1) 
\thicklines
\put(0,0){\line(1,2){5}}
\put(0,0){\line(-1,2){5}}
\put(0,0){\circle*{2}}
\put(5,10){\circle*{2}}
\put(-5,10){\circle*{2}}
\end{picture}\,\,}

\def\diamond{\,\,\begin{picture}(0,1) 
\thicklines
\put(0,0){\line(1,2){5}}
\put(0,0){\line(-1,2){5}}
\put(0,0){\circle*{2}}
\put(0,20){\circle*{2}}
\put(5,10){\circle*{2}}
\put(-5,10){\circle*{2}}
\put(-5,10){\line(1,2){5}}
\put(5,10){\line(-1,2){5}}
\end{picture}\,\,}

\def\topvee{\,\,\begin{picture}(0,1) 
\thicklines
\put(0,10){\line(1,2){5}}
\put(0,10){\line(-1,2){5}}
\put(0,10){\circle*{2}}
\put(5,20){\circle*{2}}
\put(-5,20){\circle*{2}}
\put(0,0){\circle*{2}}
\put(0,0){\line(0,1){10}}
\end{picture}\,\,}

\begin{document}

\maketitle

\abstract{We explore whether the growth dynamics paradigm of Causal Set Theory is compatible with past-infinite causal sets. We modify the Classical Sequential Growth dynamics of Rideout and Sorkin to accommodate  growth ``into the past'' and discuss what form physical constraints such as causality could take in this new framework. We propose convex-suborders as the ``observables'' or ``physical properties'' in a theory in which causal sets can be  past-infinite and use this proposal to construct a manifestly covariant framework for dynamical models of growth for  past-infinite causal sets.}


\newpage
\tableofcontents

\newpage

\section{Introduction}

Much of the effort directed towards obtaining a dynamics for Causal Set Theory has been guided by the paradigm of \textit{growth} in which a causal set grows via a stochastic process of accretion of spacetime atoms.\footnote{The other main avenue is to construct a ``quantum state sum'' over causal sets each weighted by an amplitude, for example \cite{Brightwell:2007aq,Surya:2011du,Glaser:2017sbe,Loomis:2017jhn,Cunningham:2019rob}.}  After the pioneering work by Rideout and Sorkin \cite{Rideout:1999ub},  work has concentrated on the classical domain -- e.g. \cite{Rideout:phd, Rideout:2000fh, Brightwell:2002yu,Brightwell:2002vw, Brightwell:2009x} -- though work has also been done on investigating how quantum growth models might be constructed \cite{Criscuolo:1998gd,Brightwell:2009x, Sorkin:2012xx,Dowker:2010qh,Surya:2020cfm}.

The archetypal and to date most fruitful and most studied growth dynamics for causal sets is Rideout and Sorkin's family of Classical Sequential Growth (CSG) models \cite{Rideout:1999ub}. In each of these models, a single element is born at each stage and an infinite random causal set is grown when the process is run to infinity. A CSG model is constrained by the requirement of ``internal temporality'' namely that that at each stage of the process, the new element cannot be born to the past of --- cannot \textit{precede} in the causal set order --- an element born at an earlier stage. This internal temporality constraint on the process fixes the sample space of the CSG model: it is the set of infinite \textit{past-finite} causal sets, where the term past-finite will be precisely defined shortly. Essentially, the causal set universe grown in a CSG model must have a beginning, by definition of the model. As such, the CSG models rule out the possibility that there might, for example, have been an infinite sequence of epochs in a bouncing scenario, punctuated by infinitely many ``Big Crunch-and-then-Big Bang'' events,   prior to our present epoch. 

In this work, we consider whether the growth dynamics paradigm necessarily entails past-finiteness or whether it can be compatible with \textit{past-infinite} causal set cosmologies as suggested by W\"uthrich and Callender \cite{Wuthrich:2015vva}.  In particular, we will investigate causal set cosmologies which are both past-infinite and future-infinite, \textit{i.e.} cosmologies in which time has neither a beginning nor an end. After setting out notation and concepts in section 2, in section  \ref{section_seq_growth} we modify the CSG models to accommodate growth of such causal sets. Already at this point, conceptual challenges arise, as might be anticipated. Perhaps the most pressing of these is that our new framework requires that new elements be born to the past of existing ones, thus making it (nearly if not entirely) impossible to conceive of the growth process as a physical process of Becoming \cite{Sorkin:2007qh, spacetimeatoms}. 
Nevertheless, we are able to identify a set of meaningful, comprehensible observables\footnote{We use the term ``observable'' as a shorthand for ``physical property'' and not to imply that there need be any external observer.}  for past-infinite growth dynamics, namely the \textit{convex-events} that specify which convex-suborders  are contained in the growing causal set. This sets the stage for section \ref{section_cov_growth} where we pursue an alternative route to past-infinite growth by constructing a variation of \textit{covtree} which is the basis of a manifestly covariant alternative to the framework of sequential growth models \cite{Dowker:2019qiz}. We show that the resulting framework is compatible with past-infinite growth and that the observables in this case are exactly the formerly identified convex-events. We conclude with a discussion in section \ref{section_discussion}.

\section{Preliminaries}

In this section we present terminology and notation that we use in the rest of this work, beginning with some standard terminology. 

\vspace{2mm}\noindent Let $\Pi$ be a countable (finite or infinite)  causal set (or ``causet'' for short). We adopt the irreflexive convention for the relation on $\Pi$: $x \not\prec x$, $x\in \Pi$. Recall that a causal set is locally finite by definition: $| \{ z\,|\, x \prec z \prec y\} | < \infty $ $\forall x, y \in \Pi$ such that $x\prec y$. 

\vspace{2mm}\noindent The \textbf{past} of $x\in \Pi$ is the subcauset $past(x):=\{y\in \Pi | y\prec x\}$. This is the \textit{non-inclusive} past, \textit{i.e} $x \not\in past(x)$. The \textbf{future} of $x\in \Pi$ is the subcauset $future(x):=\{y\in \Pi | y\succ x\}$. This is the \textit{non-inclusive} future.

\vspace{2mm}\noindent $\Pi$ is \textbf{past-finite} if $|past(x)|<\infty\, \ \forall \ x \in \Pi$. Similarly, $\Pi$ is \textbf{future-finite} if $|future(x)|<\infty\, \ \forall \ x \in \Pi$.

\vspace{2mm}\noindent $\Pi$ is \textbf{past-infinite} (future-infinite) if it is not past-finite (future-finite).

\vspace{2mm}\noindent $\Pi$ is \textbf{two-way infinite} if it is both past-infinite and future-infinite. Building growth dynamics for two-way infinite causet cosmologies is the motivation for this current work.

\vspace{2mm}\noindent A \textbf{stem} in $\Pi$ is a finite subcauset $\Phi$ of $\Pi$ such that if $x\in \Phi$ then $past(x)\subseteq\Phi$. An $n$-stem is a stem with cardinality $n$.

\vspace{2mm}\noindent If $\Pi$ is past-finite then an element $x\in \Pi$ is in \textbf{level $L$} in $\Pi$ if the longest chain of which $x$ is the maximal
element has cardinality $L$, \textit{e.g.} level 1 comprises the minimal elements of $\Pi$. 

\vspace{2mm}\noindent The \textbf{width} of $\Pi$, $w(\Pi)$, is the largest cardinality of an antichain in $\Pi$. The \textbf{height} of $\Pi$, $h(\Pi)$, is largest cardinality of a chain in $\Pi$. 
If $\Pi$ is past finite, the height of $\Pi$  equals the number of levels in $\Pi$. Note, the
height and width may be infinite if $\Pi$ is infinite. 

\vspace{2mm}\noindent A \textbf{path} in $\Pi$ is a (finite or infinite) chain in $\Pi$ such that the relation between each adjacent pair of elements in the chain is a link (\textit{i.e.} a covering relation) in $\Pi$. 

\subsection{Natural labelings and labeled causets}\label{subsec_labeled_causets_chap1}
Labeled causets as defined below are used throughout this paper. We emphasise that the definition of labeled causets which we give here is different to that given in \cite{Dowker:2019qiz}---it is an extension that allows us to discuss past-infinite causets. Correspondingly, definitions deriving from labeled causets (\textit{e.g.} the definition of an $n$-order) and the symbols we use to denote spaces of labeled causal sets (\textit{e.g.} $\lc{\Omega}(n)$ and $\Omega$) take a different meaning here to that in \cite{Dowker:2019qiz, Brightwell:2002vw}. 

\vspace{2mm}\noindent Let $\Psi$ be a countably infinite causet. Let $\mathbb{Z}^-$ be the set of negative integers.

\vspace{2mm}\noindent A \textbf{natural labeling} of $\Psi$ is a bijection $f$ from either $\mathbb{N}$ or $\mathbb{Z}^-$ or $\mathbb{Z}$ to $\Psi$ that satisfies $f(i) \prec f(j)\implies i < j$. 

\vspace{2mm}\noindent
The following lemma will be useful:

\begin{lemma}\label{labeling_lemma} 

Let $\Psi$ be a countably infinite causet. Then, 
\begin{itemize}
 \item[(a)] $\Psi$ has a natural labeling by $\mathbb{N}$ if and only if $\Psi$ is past-finite \cite{Brightwell:2011};
 \item [(b)]  $\Psi$ has a natural labeling by $\mathbb{Z}^-$ if and only if $\Psi$ is future-finite (a corollary of (a));
  \item[(c)] $\Psi$ has a natural labeling by $\mathbb{Z}$ if and only if one of the following conditions holds
  \cite{honan:2018,Gupta:2018}: 
  \begin{enumerate} \item[(i)] $\Psi$ is two-way infinite;
   \item[(ii)] $\Psi$ is past-finite and has infinitely many minimal elements; \item[(iii)] $\Psi$ is future-finite and has infinitely many maximal elements.\end{enumerate}
\end{itemize}
\end{lemma}

\noindent Note that cases $(c)(ii)$ and $(c)(iii)$ are each disjoint from $(c)(i)$ but not from each other, \textit{e.g.} the infinite antichain satisfies $(c)(ii)$ and $(c)(iii)$.

\vspace{2mm}\noindent  For any pair of integers $k\leq l$, let $[k,l]$ denote the set of integers $\{k,k+1,\dots,l-1,l\}$. Let $\Pi_n$ be a finite causet of cardinality $n$. 

\vspace{2mm}\noindent A \textbf{natural labeling} of $\Pi_n$ is a bijection $f:[k,k + n -1] \rightarrow \Pi$ that satisfies $f(i) \prec f(j)\implies i < j \ \ \forall \  i,j \in [k,k+n-1]$, where $k\in \mathbb{Z}$. 

\vspace{2mm}\noindent  A \textbf{finite labeled causet} is a causet with ground-set $[k,l]$,  where $k \le l$, whose order satisfies the condition: $x\prec y \implies x<y$, \textit{i.e.} it is a causet for which the identity map is a natural labeling (hence its name).

\vspace{2mm}\noindent An \textbf{infinite labeled causet} is a causet with ground-set $\mathbb{N}$ or $\mathbb{Z}^-$ or $\mathbb{Z}$ whose order satisfies the condition: $x\prec y \implies x<y$ (as in the finite case, it is a causet for which the identity map is a natural labeling).

\vspace{2mm}\noindent 
From now on we will denote labeled causets and their subcausets by capital Roman letters with a tilde, \textit{e.g.} $\lc{C}$. We often (but not always) use a subscript to denote the cardinality of a labeled causet, \textit{e.g.} ``$\lc{C}_n$ has cardinality $n$''. 

\vspace{2mm}\noindent 
Given some $n\in\mathbb{N}^+$, we denote the set of all labeled causets with cardinality $n$ by $\tilde{\Omega}(n)$.
Note that given a labeled causal set $\lc{C}_n \in \tilde{\Omega}(n)$ with ground set $[0, n-1]$, for each integer $k$ there is an isomorphic labeled causet with ground set $[k, n-1+k]$ that is gotten from $\lc{C}_n$ by adding $k$ to each of its elements. Therefore there are infinitely many labeled causets in $\tilde{\Omega}(n)$. This is not the case in 
previous works on past-finite causet growth models where the ground set of a finite labeled causet of cardinality $n$ is fixed to be $[0,n-1]$. 

\vspace{2mm}\noindent The set of all infinite labeled causets whose ground set is $\mathbb{N}$, $\mathbb{Z}^-$ or $\mathbb{Z}$ respectively is denoted by $\tilde{\Omega}_{\mathbb{N}},\tilde{\Omega}_{\mathbb{Z}^-}$ or $\tilde{\Omega}_{\mathbb{Z}}$, respectively.\footnote{ A note of caution: in previous works on past-finite causet growth models, the notation $\tilde{\Omega}_{\mathbb{N}}$ has been used for the set of all finite labeled causal sets.}

\vspace{2mm}\noindent 
The set of all infinite labeled causets is denoted by $\tilde{\Omega}\equiv \tilde{\Omega}_{\mathbb{N}}\sqcup\tilde{\Omega}_{\mathbb{Z}^-}\sqcup\tilde{\Omega}_{\mathbb{Z}}$.

\vspace{2mm}\noindent
A CSG model \cite{Rideout:1999ub} grows past-finite causal sets, \textit{i.e.} its sample space is $\lc{\Omega}_{\mathbb{N}}$. 

\subsection{Orders}\label{subsec_orders}

We write $\lc{C}\cong\lc{D}$ if  labeled causets $\lc{C}$ and $\lc{D}$ are equal up to an order-isomorphism.

\vspace{2mm}\noindent An \textbf{order}, $C$, is an order-isomorphism class of labeled causets. We denote orders by capital Roman letters without a tilde.

\vspace{2mm}\noindent Given an order $C$, its cardinality $|C|$ is defined to be the cardinality of a representative of $C$. Similarly, the width and height of an order are those of its representatives. An order is future-finite if its representatives are future-finite \textit{etc.} When we refer to  elements of $C$, we mean elements of a representative of $C$ and the meaning should be clear from the context as in for example: ``$C$ has 5 minimal elements.''

\vspace{2mm}\noindent An \textbf{$n$-order} is an order with cardinality $n$.

\vspace{2mm}\noindent For each $n\in\mathbb{N}$, $\Omega(n)$ denotes the set of $n$-orders. $\Omega(n)$ is a finite set. 

\vspace{2mm}\noindent $\Omega:= \lc{\Omega}/\cong$ is the set of infinite orders.

\vspace{2mm}\noindent $\Omega_{\mathbb{Z}},\Omega_{\mathbb{Z}^-}$ and $\Omega_{\mathbb{N}}$ are the subsets of $\Omega$ that have a representative labeled by $\mathbb{Z},\mathbb{Z}^-$ and $\mathbb{N}$, respectively.

\vspace{2mm}\noindent  Note that $\Omega= \Omega_{\mathbb{Z}}\cup\Omega_{\mathbb{Z}^-}\cup\Omega_{\mathbb{N}}$. By lemma  \ref{labeling_lemma}, the union is not disjoint. $\Omega_{\mathbb{Z}}\cap\Omega_{\mathbb{Z}^-}\cap\Omega_{\mathbb{N}}=\Omega_{\mathbb{Z}^-}\cap\Omega_{\mathbb{N}}$ is the set of past-and-future-finite orders that have infinitely many maximal elements and infinitely many minimal elements and is nonempty: the union of infinitely many disjoint 2-chains for example.  $\Omega_{\mathbb{Z}}\cap\Omega_{\mathbb{Z}^-}$ is the set of future-finite orders that have infinitely many maximal elements. $\Omega_{\mathbb{Z}}\cap\Omega_{\mathbb{N}}$ is the set of past-finite orders that have infinitely many minimal elements. 

\subsection{Convex-suborders}
Let $\Pi$ and $\Psi$ be causal sets.

\vspace{2mm}\noindent 
$\Pi$ is a \textbf{convex-subcauset} in $\Psi$ if $\Pi$ is finite and $\Pi\subseteq \Psi$ and, whenever $x,y\in \Pi$ and $x\prec z \prec y$ in $\Psi$, then $z\in\Pi$.

\vspace{2mm}\noindent We say that $\Psi$ contains a \textbf{copy} of $\Pi$ if there exists a convex-subcauset $\Pi'\subseteq \Psi$ that is order-isomorphic to $\Pi$. 

\vspace{2mm}\noindent 
Let $C$ and $D$ be orders with (arbitrary) representatives $\lc{C}$ and $\lc{D}$, respectively.

\vspace{2mm}\noindent We say that $C$ is a \textbf{convex-suborder} in $D$ if $\lc{D}$ contains a copy of $\lc{C}$. 
Note that this definition is independent of the representatives $\lc{C}$ and $\lc{D}$ because the definition of ``contains a copy of'' is less restrictive than ``contains as a subcauset''.   In that case we also say that $C$ is a convex-suborder in $\lc{D}$. If the cardinality of convex-suborder $C$ equals $n$ we say that $C$ is an \textbf{$n$-convex-suborder} in $D$ or in $\lc{D}$.

\vspace{2mm}\noindent We say that an order $C$ is a \textbf{convex-rogue} if there exists another order $D$ that is not isomorphic to $C$ and that has the same  convex-suborders as $C$. In that case we say that $C$ and $D$ are a \textbf{convex-rogue pair}.\footnote{This terminology follows that of \cite{Brightwell:2002vw} in which a pair of rogues are two past-finite, non-isomorphic orders with the same stems.}

Note that the convex-subcausets (convex-suborders) are ordered by inclusion. If $A$ is a convex-subcauset of $B$ and $B$ is a convex-subcauset of $C$, then $A$ is a convex-subcauset of $C$. 

\section{Sequential growth}\label{section_seq_growth}

The paradigm of growth dynamics is motivated by the heuristic concept of \textit{Becoming}: the discrete causal set spacetime comes into being \textit{ex nihilo} via an unceasing process of the birth of causal set elements. While the concept of Becoming could be regarded simply as a crutch in defining a model of random infinite causal sets and dispensed with once a measure has been defined on the full sigma algebra of events, Sorkin has proposed that growth is a physical process in which the birth of an element is the \textit{happening} of an event, while the element itself signifies that the event (of its birth) has \textit{already happened} \cite{Sorkin:2007qh, spacetimeatoms}. This viewpoint allows the passage of time to be manifested within physics as the growth of a causal set.

Perhaps the most intuitive notion of growth is that of ``sequential growth'' in which the causal set grows through a sequential accretion of elements, somewhat akin to a tree growing at the tips of its branches. A sequential growth process for causal sets is made up of stages, labeled by the natural numbers, a discrete parameter. Starting at stage 0, at each stage $n$ in the sequence a new element  is born. The new element is born with randomly chosen relations with the already existing elements according to a model-dependent probability distribution. So, at the end of  stage $n$, the growing, partial causet contains $n+1$ elements. In the limit $n\rightarrow \infty$, the process generates an infinite causal set.

The Classical Sequential Growth (CSG) models are the archetype of sequential growth models. First introduced in \cite{Rideout:1999ub}, the CSG models have proved to be a fruitful arena for studying causal set cosmology \cite{Martin:2000js,Brightwell:2002vw,Ash:2002un,Varadarajan:2005gg,Dowker:2017zqj} and for developing new dynamical frameworks \cite{Dowker:2019qiz,Surya:2020cfm,Zalel:2020oyf}. Though the CSG models themselves do not generate past-infinite causal sets,  they are a natural starting point for trying to construct dynamics for two-way infinite causal sets.

\subsection{Alternating growth}
As mentioned, the CSG models themselves do not generate past-infinite causal sets. This is not a probabilistic statement:  there are no past-infinite causets at all in the sample space for the process. Each CSG model satisfies a condition known as Internal Temporality which states that at each stage the new element cannot be born to the past of -- cannot precede in the causet order -- an existing element. Indeed, the first challenge in generalising the CSG models to the past-infinite case is generalising the condition of Internal Temporality. If we are to both generate past-infinite causal sets and keep the essence of sequential growth -- \textit{i.e.}, that starting from the empty set, new elements are born in a sequence of stages -- we must loosen the condition of Internal Temporality to allow elements to be born to the past of existing elements. This move breaks the compatibility between the label of the stage of the sequential growth process, the concept of the birth of the element as manifesting the physical happening of the event and the order of the resulting causet as being the physical order -- before and after --  in which the elements are born. Nevertheless, mathematically at least, there is a way to generalise the condition of Internal Temporality that keeps some of its power. 

In a CSG model,  $\mathbb{N}$ is the ground set of the growing causal set, and at the stage labeled $n$ the element $n$ is born. In this context, Internal Temporality is equivalent to the requirement that the growing causal set is naturally labeled by $\mathbb{N}$ -- \textit{i.e.} the sample space of the growth process is $\lc{\Omega}_{\mathbb{N}}$ and the growth process can be conceived of as a random walk up ``labeled poscau.''\footnote{``Poscau'' is short for the ``partial order of causal sets'', and ``labeled'' signifies that the causal sets in the order are labeled causets.}
\begin{definition}\label{def_lab_poscau}\par Labeled poscau is the partial order on the set of finite labeled causets whose ground set is $[0,n]$,  for all $n\in\mathbb{N}$, where $\lc{S}\prec \lc{R}$ if and only if $\lc{S}$ is a stem in $\lc{R}$.\footnote{We use the symbol $\prec$ to denote the relation for several different partial orders in this work. The meaning of $\prec$ in each case is to be inferred from the context.}
\end{definition}

\noindent Labeled poscau is a rooted directed tree and its first three levels  are shown in figure \ref{lposcau}.
\begin{figure}[htpb]
  \centering
	\includegraphics[width=0.6\textwidth]{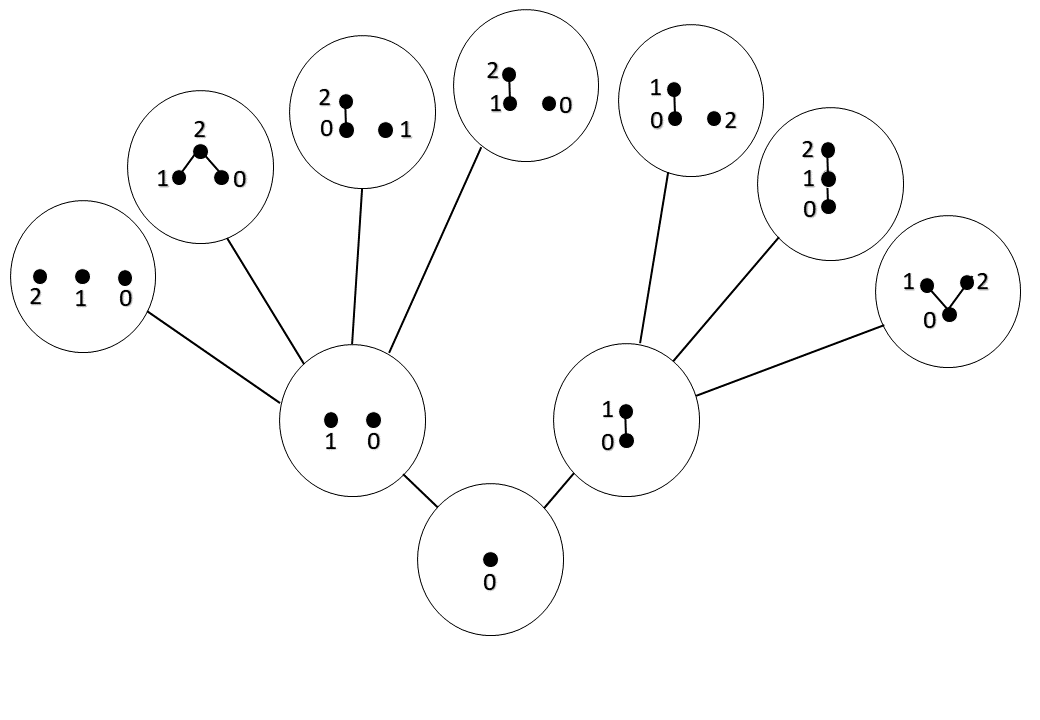}
	\caption{The first three levels of labeled poscau.}
	\label{lposcau}
\end{figure}

Reformulating Internal Temporality as a statement about natural labelings reveals a candidate generalisation of it to the two-way infinite case, namely that the infinite causal set that is grown has a natural labeling by $\mathbb{Z}$: it is an element of $\lc{\Omega}_{\mathbb{Z}}$~. Some freedom remains in how to translate this condition back into a statement about the sequence of the birth of the elements of the causet. 

For definiteness, in this work we fix the freedom thus:  let the positive and negative integers be born in an alternating sequence, $0,-1,1,-2,2...$, so that at stage $n$, if $n$ is even the element $\frac{n}{2}$ is born and if $n$ is odd the element $-\frac{n+1}{2}$ is born. We call a transition in which a positive element is born a ``forward transition''. Similarly, a ``backward transition'' is one in which a negative element is born, so a transition $\lc{C}_n\rightarrow \lc{C}_{n+1}$ is forward when $n$ is even and backward when $n$ is odd. Note that, in this framework of alternating growth, the natural number label of the stage is not equal to the element of the causet born at that stage (as it is in CSG models) though it is still the case that the label of the stage equals the cardinality of the partial causet at the beginning of the stage (as it is in CSG models). 

Internal Temporality in this context becomes the condition: positive elements cannot be born to the past of  elements born at previous stages, negative elements cannot be born to the future of elements born at previous stages. In particular, at no stage can an element be born between two elements that were born at previous stages. This implies that at each stage,  the finite partial causet is a convex-subcauset of the partial causet at the next stage, and thence of the infinite causet that is the union of all the partial causal sets at all the infinitely many stages. 

We dub the resulting dynamical framework ``alternating growth''. The alternating growth process can be represented as a random walk up ``alternating poscau'', a directed rooted tree whose nodes are finite labeled causets. More precisely,
\begin{definition}\label{def_alt_poscau}\par 

Alternating poscau is the partial order on the set of finite labeled causets whose ground set is $[-n,n]$ or $[-n,n+1]$ for all $n>0$, where $\lc{S}\prec \lc{R}$ if and only if $\lc{S}$ is a convex-subcauset in $\lc{R}$.
\end{definition}
The first three levels of alternating poscau are shown in figure \ref{alt_poscau}.\footnote{There is a bijection, $f$ from the set of nodes at level $k$ in labeled poscau to the set of nodes at level $k$ in alternating poscau where $f$ takes a labeled causet and maps each element $x$ to $x - \lfloor{k/2}\rfloor $. Note however that $f$ is not an isomorphism between labeled poscau and alternating poscau.} Note that the levels of alternating poscau are finite because of the restriction on the ground sets of the finite labeled causets to $[-n,n]$ or $[-n,n+1]$.

 There is a bijection from the set of infinite paths starting at the root in alternating poscau to $\lc{\Omega}_{\mathbb{Z}}$, where an infinite path $\lc{C}_1\prec\lc{C}_2\prec...$ maps to $\lc{C}=\bigcup_{n>0} \lc{C}_n$. The standard technology of stochastic processes and measure theory then provides the $\sigma$-algebra of measurable events  generated by the semi-ring of all cylinder sets, each associated with a node of alternating poscau: $cyl(\lc{C}_n)\subset \lc{\Omega}_{\mathbb{Z}}$ is the set of labeled causets on ground set $\mathbb{Z}$  that contain $\lc{C}_n$ as a convex-subcauset.  A random walk on alternating poscau specified in terms of transition probabilities corresponds to a unique measure on this measurable space and, vice versa, every measure on the $\sigma$-algebra 
 generated by the cylinder sets gives a unique collection of transition probabilities for every transition. 
 
By re-interpreting the Internal Temporality condition as above, we are thus able to modify the sequential growth paradigm to allow growth of two-way infinite causets.\footnote{One can consider growth models with different rules, leading to different trees: we refer the reader to \cite{Bruno:2018} for such variations, \textit{e.g.} sequential growth models in which the decision to make a forward or a backward transition at each stage is random.} Recall, however,  that by lemma \ref{labeling_lemma} the set of two-way infinite causets (case $(c)(i)$ in lemma \ref{labeling_lemma})  is a proper subset of the sample space ${\lc{\Omega}}_{\Z}$. It turns out that the set of two-way infinite causets is a measureable set and therefore it will be up to the dynamics (\textit{i.e.}, the specific random walk) whether the set of two-way infinite causets has measure one or not. Indeed, one can ask whether one can identify conditions on the transition probabilities that will imply that the causet will almost surely be two-way infinite. 

\begin{lemma}\label{lemma_two-way}  Let $W$ be the set of two-way infinite labeled causets. $W$ is a measureable set in an alternating growth dynamics, i.e. a random walk up alternating poscau. 

\end{lemma}

\begin{proof}
$W = W^+ \cap W^-$ where $W^+$ ($W^-$) is the set of causets in $\lc{\Omega}_{\mathbb{Z}}$ that have an element  with an infinite future (past). We will show that $W^+$ is measureable and the proof for  $W^-$  is similar. 

For each integer $k \in \mathbb{Z} $ let  $\Gamma_k$ be the set of  causets in $\lc{\Omega}_{\mathbb{Z}}$ such that the element $k$  has an infinite future. $W^+$ is the union of all the $\Gamma_k$. 

For each $k\in \mathbb{Z}$, $m, n \in \mathbb{N}$ s.t. $m>0$ and $n>|k|+m$ let  $\lc{\Omega}_{k, n, m}$ be the set of finite labeled causets on the ground set $[-n,n]$ such that there are $m$ elements above element $k$. Take the union over the set $\lc{\Omega}_{k, n, m}$ of all the associated cylinder sets and call that union $\Gamma_{k, n, m}$:
\begin{align}
\Gamma_{k,n,m} : =   \bigcup_{ \lc{C} \in \lc{\Omega}_{k, n, m}}  cyl(\lc{C})\,.
\end{align}

 Then 
\begin{align}
\Gamma_k =  \bigcap_{m=1}^\infty\bigcup_{n=| k |+m +1}^\infty  \Gamma_{k, n, m}\,.
\end{align}

\end{proof}

\begin{figure}[h]
\centering
    \includegraphics[scale=0.37]{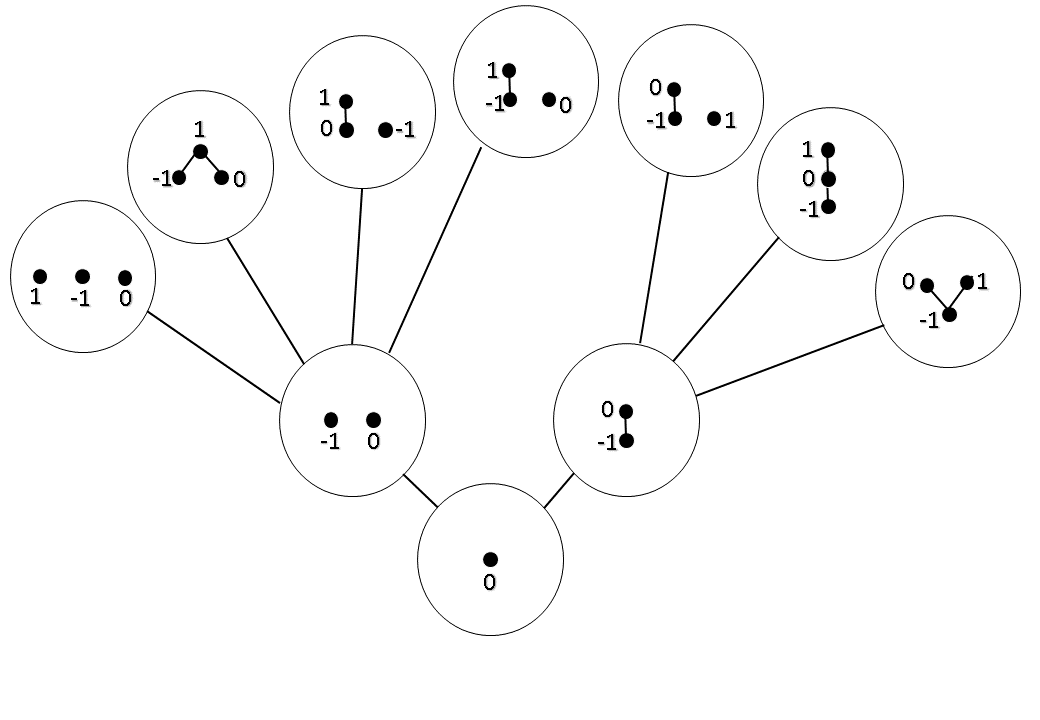}
\caption{The first three levels of alternating poscau.}
\label{alt_poscau}
\end{figure}
\subsection{Alternating growth dynamics}
While any random walk on alternating poscau gives rise to a well-defined measure space with sample space $\lc{\Omega}_{\mathbb{Z}}$, not every such walk will be interesting physically and it remains for us to identify classes of interest. 

This is completely analogous to the past-finite case, where the CSG models were identified as a physically-meaningful subclass of the random walks on labeled poscau. Indeed, the CSG models are exactly the random walks on labeled poscau that satisfy the physically motivated conditions of Discrete General Covariance, and Bell Causality, to be discussed further below. These conditions were solved and the transition probabilities in a CSG model proved to take the following form:
\begin{align}\label{transprob}
\mathbb{P}(\tilde{C}_n \rightarrow \tilde{C}_{n+1}) &= \frac{\lambda(\varpi, m)}{\lambda(n, 0)}\,,
\end{align}
where $\mathbb{P}(\lc{C}_n\rightarrow \lc{C}_{n+1})$ is the probability of transition from  $\lc{C}_n$ to one of its children, $\lc{C}_{n+1}$; $\varpi$ and $m$ are the number of new relations and new links, respectively, formed with the newborn element  at stage $n$; and the function $\lambda$ is given by,
\begin{equation} \label{weight_sec_1}
\lambda(k, p) := \sum_{i = 0}^{k - p} \binom{k-p}{i} t_{p+i},
\end{equation}
where $\{t_0,t_1,t_2,...\}$ is an infinite set of real non-negative parameters or ``couplings'' (with $t_0>0$) that specify the particular CSG model. As the transition probabilities are ratios of linear combinations of the $t_n$'s, there is a (projective) equivalence relation on the sets $\{t_n\}$, which freedom can be fixed by setting $t_0 = 1$. 

Note that above we referred to the ``new relations and new links [...] formed with the newborn element  at stage $n$'' without mentioning that the newborn element is $n$ in a CSG model and without mentioning that in any
new relation the newborn element must succeed (be above) the element born at a previous stage.  We have made these omissions because, by doing so,  we can adopt 
equation \eqref{transprob} and its succeeding text definition \textit{as is}  for the definition of an alternating growth model simply by letting $\lc{C}_n$ and $\lc{C}_{n+1}$ denote nodes in alternating poscau such that  $\lc{C}_n\rightarrow \lc{C}_{n+1}$ is a possible transition in alternating poscau.  Now, however, when the stage label $n$ is odd the transition is a backward transition and in any new relation the newborn element must precede (be below) the already existing element. We call this new family of models the family of Alternating CSG dynamics \cite{honan:2018,Gupta:2018}. Given a CSG model with parameters $\{t_0,t_1,t_2,...\}$, its alternating counterpart is the Alternating CSG model with the same set of parameters. 

Do the Alternating CSG dynamics retain any of the features that make CSG models physically interesting? For example, do the Alternating CSG models satisfy any sort of causality condition? 
In the remainder of this section we identify the form that four key attributes---covariance, causality, causal immortality and meaningful observables---might take in the alternating growth framework and discuss whether the Alternating CSG models possess these attributes.

Before turning to the question of physical conditions, we introduce the example of the most well-studied family of CSG models, Transitive Percolation---a 1-parameter family of CSG models given by $t_k=t^k$ where $t$ is a positive real constant \cite{Alon:1994, Rideout:1999ub}. Its alternating growth counterpart, Alternating Transitive Percolation, is defined by the same couplings: $t_k=t^k$ for some $t>0$. For Transitive Percolation, the transition probability given in equation \eqref{transprob} takes the simple form,
\begin{equation}\label{atp_020321}\mathbb{P}(\tilde{C}_n \rightarrow \tilde{C}_{n+1}) =p^mq^{n-\varpi},\end{equation}
where  $p=\frac{t}{1+t}$ and $q=1-p$, so that $p \neq 1$ and $p \neq 0$ (and as before, $\varpi$ and $m$ are the numbers of new relations and new links, respectively, formed with the element that is born at stage $n$). The interpretation of equation \eqref{atp_020321} is that the new element born in the transition forms a relation with each existing element with probability $p$ independently and then the transitive closure is taken to obtain $\lc{C}_{n+1}$.
With this interpretation, $t_k$ is the relative probability that the new element forms exactly $k$ relations (before taking the transitive closure). Equation \eqref{atp_020321} and the functional form of the couplings $t_k=t^k$ reflect the ``local'' nature of Transitive Percolation. All other CSG models can be seen as ``non-local'' generalisations of Transitive Percolation in which the probability of the newborn forming a relation with a given element depends on whether or not relations are formed with the other existing elements. 

Equation \eqref{atp_020321} and hence this form of ``locality'' is retained by Alternating Transitive Percolation, so that there are close similarities between the two models. For example, let $\lc{C}_n$ and $\lc{D}_n$ be nodes in labeled poscau and alternating poscau respectively, and let $\lc{C}_n\cong\lc{D}_n$. Then \cite{honan:2018,Gupta:2018}, 
\begin{lemma} 
The probability of reaching $\lc{C}_n$ in a particular CSG dynamics is equal to the probability of reaching $\lc{D}_n$ in the alternating growth counterpart of that CSG dynamics
 if and only if the CSG dynamics in question is Transitive Percolation. 
\end{lemma} 
Thus, at any finite stage of growth, Transitive Percolation (TP) and Alternating Transitive Percolation (ATP) cannot be distinguished. However the processes are different when run to infinity. For example, in TP,  the past finite causet grown almost surely  has infinitely many posts i.e. infinitely many elements $\{ k_1, k_2, \dots\}$ such that $0\le k_1 < k_2 < k_3 \dots $ and every element of the causet is related to all the $k_i$. In ATP, almost surely a two-way infinite causet is grown in which there are again infinitely many posts in the past and in the future: elements $\{ \dots k_{-2}, k_{-1}, k_{0}, k_{1}, k_{2} \dots\} $ such that $\dots k_{-2} < k_{-1} <k_0 < k_1 < k_2  \dots $ and such that every element of the causet is related to all the $k_i$. 
TP realises the heuristic of a bouncing universe with a beginning and ATP realises the heuristic of a bouncing universe with no beginning. 

\paragraph{Covariance:} It is a tenet of causal set theory that the atoms of spacetime have no structure; it is of no physical relevance what mathematical objects the elements of a causal set are.\footnote{Our choice of labeled causets -- with their ground sets of integers -- for our world of discourse in this paper  is purely for convenience.} Only the cardinality of the causet and the order relation are physical. This implies that the mathematical identity of and labels of the causet elements are not physical and one can consider this an analogue of the ``coordinate invariance'' or ``general covariance'' of continuum General Relativity. 

In a CSG model of past-finite sequential growth, this label invariance  is manifested thus: given a pair of order-isomorphic finite labeled causets, $\lc{C}_n$ and $\lc{C}'_n$, which are nodes in labeled poscau, the probability of reaching  $\lc{C}_n$ is equal to the probability of reaching $\lc{C}'_n$, that is,
\begin{equation}\label{dgc_sec1}
\lc{C}_n\cong\lc{C}'_n\implies\mathbb{P}(\lc{C}_n)=\mathbb{P}(\lc{C}'_n).
\end{equation}
Condition \eqref{dgc_sec1} is known as Discrete General Covariance (DGC) and it can be generalised to pertain to the alternating sequential growth framework simply by letting $\lc{C}_n$ and $\lc{C}'_n$ in equation (\ref{dgc_sec1}) denote isomorphic nodes in alternating poscau.

Every CSG model satisfies the DGC condition. In contrast, the only Alternating CSG dynamics which satisfies discrete general covariance is Alternating Transitive Percolation:
\begin{claim}\label{alt_tp_cov} An Alternating CSG model satisfies the discrete general covariance condition if and only if it is an Alternating Transitive Percolation model.
\end{claim}
\begin{proof}

That Alternating Transitive Percolation satisfies DGC follows from equation \eqref{atp_020321} since it implies that the probability of reaching some $\lc{C}_n$ in alternating poscau is  $\mathbb{P}(\lc{C}_n)=p^Lq^{\binom{n}{2}-R}$, where $L$ and $R$ are the number of links and relations in $\lc{C}_n$,  respectively. These numbers $L$ and $R$ depend only on the order-isomorphism class of $\lc{C}_n$.

Now consider an Alternating CSG model defined by parameters $\{t_n\}$. 
Consider, for $n>0$,  the $(2n+1)$-order $C$  that contains a $(2n)$-antichain of which $n$ elements have a common ancestor, as shown in figure~\ref{covariance_proof_fig}.

\begin{figure}[t]
  \centering
	\includegraphics[width=0.6\textwidth]{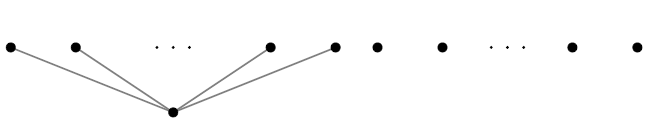}
	\caption{The $(2n+1)$-order $C$, that contains a $(2n)$-antichain of which $n$ elements have a common ancestor.}
	\label{covariance_proof_fig}
\end{figure}

\noindent Let $\lc{C}$ denote the representative of $C$ that is grown in the alternating framework in the following way: the element 0 and the elements born in the first $2n-2$ stages form an antichain, the element born at stage $2n-1$ is born to the past of $n$ of the existing elements, and the element born at stage $2n$ is unrelated to all existing elements. The probability of growing $\lc{C}$ in an Alternating CSG dynamics is $\mathbb{P}(\lc{C}) = t_0^{2n-2} t_n t_0 = t_0^{2n-1}t_n.$ 
Let $\lc{C}'$ denote another representative of $C$ that is grown in the alternating framework in the following way: the elements born in forward transitions are all born to the future of the element $0$, and the elements born in backward transitions are all born unrelated to all existing elements. The probability of growing $\lc{C}'$ in an Alternating CSG dynamics is $\mathbb{P}(\lc{C}') = t_1^nt_0^n$. If the alternating CSG model is covariant then $\mathbb{P}(\lc{C})=\mathbb{P}(\lc{C}')$, which implies that $\frac{t_1^n}{t_n}=t_0^{n-1}$. This is Alternating Transitive Percolation and can be cast into the form $t_n=t^n$ by setting $t_0=1$.

\end{proof} 

\paragraph{Causality:} Within the past-finite sequential growth framework, a dynamics is causal if it satisfies the ``Bell causality'' condition of \cite{Rideout:1999ub} which adapts the ``local causality'' condition of Bell's theorem to a causal structure that is discrete and dynamical. The Bell causality condition takes the form of an equality between ratios of transition probabilities,
\begin{equation}\label{bell_sec_1}
\frac{\mathbb{P}(\lc{C}_n\rightarrow \lc{C}_{n+1})}{\mathbb{P}(\lc{C}_n\rightarrow \lc{C}'_{n+1})}=\frac{\mathbb{P}(\lc{B}_l\rightarrow \lc{B}_{l+1})}{\mathbb{P}(\lc{B}_l\rightarrow \lc{B}'_{l+1})},
\end{equation}
where $\lc{B}_{l+1}$, $\lc{B}'_{l+1}$ and $\lc{B}_l$ are obtained from $\lc{C}_{n+1}$, $\lc{C}'_{n+1}$ and $\lc{C}_n$, respectively, by deleting one or more spectators\footnote{A spectator is an element that is spacelike to the newborn element in both transitions  $\lc{C}_n\rightarrow \lc{C}_{n+1}$ and $\lc{C}_n\rightarrow \lc{C}'_{n+1}$.} and then relabeling the remaining elements consistently. One concrete way to do this relabeling after deletion of spectators from $\lc{C}_n$ is to shift all the labels down, filling in the gaps without changing the total order, as necessary until the ground set is $[0,l-1]$: this is then $\lc{B}_l$. An example is shown in figure \ref{BellCauFig}.  For a model that satisfies DGC,  the algorithm for consistent relabeling after deletion of spectators plays no real role because the transition probabilities do not depend on the labeling, only the order-isomorphism class of the causets.  So we can say that (\ref{bell_sec_1}) holds for all relabelings, but is only one independent condition when the dynamics satisfies DGC.

What form can the Bell causality condition take within the alternating growth framework? While at first glance it may seem that equation \eqref{bell_sec_1} can be adapted to the alternating framework simply by letting $\lc{B}_{l+1}$, $\lc{B}'_{l+1}$, $\lc{B}_l$, $\lc{C}_{n+1}$, $\lc{C}'_{n+1}$ and $\lc{C}_n$ denote nodes in alternating poscau, this is not so. To see this, let $\lc{C}_n$ be a node in alternating poscau, and let $\lc{C}_{n+1}$ and $\lc{C}'_{n+1}$ denote two of its children. Now, construct $\lc{B}_l$  from $\lc{C}_n$ by removing the spectators and relabeling. Next, remove the spectators from $\lc{C}_{n+1}$ and relabel---this is where the problem arises since there may be no relabeling that produces a child of $\lc{B}_{l}$. In particular, this failure occurs whenever the number of spectators is odd because in that case if $\lc{C}_n\rightarrow \lc{C}_{n+1}$ is a forward transition
then $\lc{B}_l\rightarrow \lc{B}_{l+1}$ must be a backward transition, which leads to a contradiction. An example is shown in figure \ref{alternating_causality}. It is in these cases that the generalisation of equation \eqref{bell_sec_1} to the alternating dynamics becomes ill-defined. Instead, we will use a weakened causality condition (in similarity to the weakened causality conditions of \cite{Varadarajan:2005gg,Dowker:2005gj}) that states that an alternating dynamics is causal if equation \eqref{bell_sec_1} is satisfied whenever there is a relabeling such that the condition is well-defined. 

Having arrived at a proposed Bell causality condition for the alternating framework, we can ask whether the Alternating CSG dynamics satisfy it, beginning with Alternating Transitive Percolation. Since ATP is covariant (as we showed in claim \ref{alt_tp_cov}), the relabeling issue in the Bell causality condition is moot and we can use equation \eqref{atp_020321} to verify that equality \eqref{bell_sec_1} is satisfied,
\begin{equation}\begin{split}
\frac{\mathbb{P}(\lc{C}_n\rightarrow \lc{C}_{n+1})}{\mathbb{P}(\lc{C}_n\rightarrow \lc{C}'_{n+1})}=\frac{p^mq^{n-\varpi}}{p^{m'}q^{n-\varpi '}}=\frac{p^mq^{l-\varpi}}{p^{m'}q^{l-\varpi '}}=\frac{\mathbb{P}(\lc{B}_l\rightarrow \lc{B}_{l+1})}{\mathbb{P}(\lc{B}_l\rightarrow \lc{B}'_{l+1})},
\end{split}
\end{equation}
where $\varpi'$ and $m'$ denote the number of relations and links, respectively, formed by the element that is born at stage $n$ in the transition $\tilde{C}_n \rightarrow \tilde{C}'_{n+1}$. In this sense, the Alternating Transitive Percolation models are ``Bell causal''. Since the remaining Alternating CSG dynamics are not covariant, to ascertain whether they are causal either requires specifying a canonical method of relabeling by which $\tilde{B}_{l+1}$ should be obtained from $\tilde{C}_{n+1}$ \textit{etc.}, which renders the Bell causality condition itself label-dependent and hence not covariant, or the condition \eqref{bell_sec_1} should be imposed for each consistent relabeling that exists. 


Having discussed formally adapting equality \eqref{bell_sec_1} to the alternating growth framework, we turn to  the question of the physical interpretation of this proposed new Bell causality condition which  is far from clear. The ``local causality'' condition in Bell's theorem captures the heuristic that the outcome of a given event can only be influenced by the events inside its past lightcone. In this spirit, within the framework of past-finite growth, the ``Bell causality'' condition (equation \eqref{bell_sec_1}) states that at each stage of the growth process, the probability for each transition depends only on the past of the new-born element. But this interpretation is obliterated in the alternating growth framework. In a forward transition, the transition probability depends only on the past of the new-born element---but not on its entire past, since some of it has not yet been determined. The situation is even worse in the backward transitions where the transition probabilities depend on the future of the new-born element. One resolution is to require that equality \eqref{bell_sec_1} holds only for the forward transitions (\textit{i.e.}, when $n$ is even), leaving the backward transitions unconstrained by causality. Or it may be that we need an altogether new way of thinking about causality in the alternating framework, if sense can be made of it at all. 

\begin{figure}[h]
  \centering
	\includegraphics[width=0.6\textwidth]{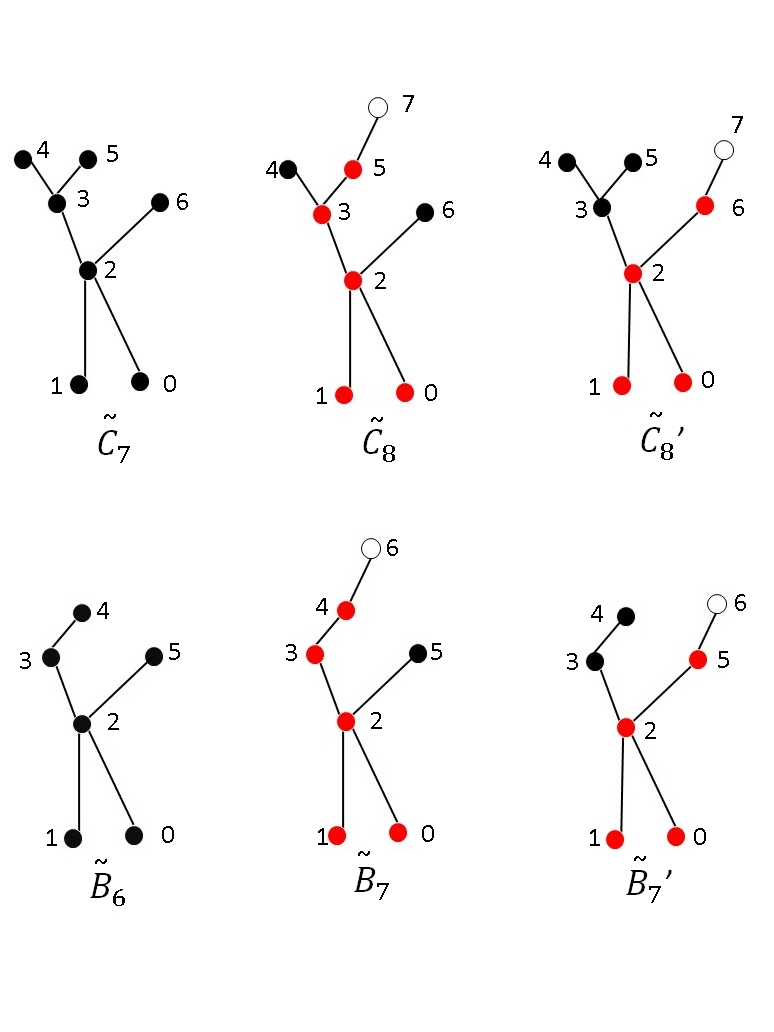}
	\caption{An illustration of the Bell Causality condition in past-finite sequential growth. The parent $\lc{C}_7$ and two of its children are shown on the top line. From these, the parent $\lc{B}_6$ and two of its children can be obtained by removing the element 4 (\textit{i.e.} the spectator) and relabeling. The new-born element in each child is shown in white. The past of the new-born element in each transition is shown in red.}
	\label{BellCauFig}
\end{figure}
\begin{figure}[h]
  \centering
	\includegraphics[width=0.5\textwidth]{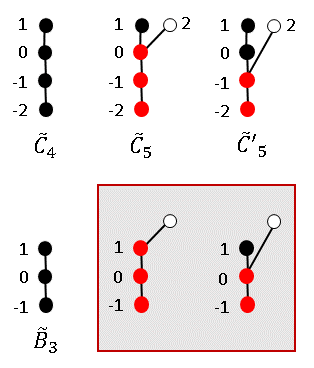}
	\caption{The Bell Causality condition is not always well-defined in the alternating growth framework. The parent $\lc{C}_4$ and two of its children are shown on the top line. The new-born element in each child is shown in white. The past of the new-born element in each transition is shown in red. $\lc{B}_3$ is constructed from $\lc{C}_4$ by removing the element 1 (\textit{i.e.} the spectator) and relabeling. Removing the spectator from $\lc{C}_5$ and $\lc{C}'_5$ results in the causets shown in the box, but there is no relabeling of these that corresponds to children of $\lc{B}_3$.}
	\label{alternating_causality}
\end{figure}

\paragraph{Causal immortality:} In past-finite sequential growth, the sample space is the space of all infinite past-finite causets, $\lc{\Omega}_{\mathbb{N}}$. This space contains a variety of cosmologies: some are future-infinite and some are future-finite, some contain infinite antichains and some do not. But in the CSG models, only a subset of all these potential configurations can be realised because the CSG models generate, with probability one, causets with no maximal elements \cite{Brightwell:2002vw}. We say that the CSG models have the property of ``causal immortality'' because the effect of each element/event reaches arbitrarily far into the future.

Similarly, in alternating sequential growth the sample space $\lc{\Omega}_{\mathbb{Z}}$ contains several causal set families (given in lemma  \ref{labeling_lemma}) but only a subset of these is realised by the Alternating CSG dynamics because these dynamics generate causets with no maximal nor minimal elements, as we show in claim \ref{immortality_claim} below. 

\begin{claim} \label{immortality_claim}
Every element in a causal set grown in an Alternating CSG model with $t_k>0$ for some $k>0$ almost surely has an element to its future and an element to its past. 
\end{claim}
\begin{proof}
Consider a growth process with an Alternating CSG dynamics with $t_k>0$ for some $k>0$. Suppose that the labeled causet $\lc{C}_n$ has been grown by the beginning of stage $n>k$, and let $x\in \lc{C}_n$ be a maximal element.

First, we show that the probability that $x$ is maximal in the complete causal set is zero. Let $r\geq n$ be an even integer. Then the probability that $x$ is maximal at the end of stage $r$ (given that $x$ is maximal at the beginning of stage $r$) is,   \begin{equation}\label{pr_def} 1-p_r= \frac{\lambda(r-1,0)}{\lambda(r,0)}.
    \end{equation}
where $p_r$ is the effective parameter of \cite{Brightwell:2016}. Therefore the probability that $x$ is maximal in  the complete causet is,
    \begin{equation}\label{eq_prob_020321}
    \begin{split}
           \lim_{s\rightarrow\infty}\mathbb{P}(x \text{ is maximal at end of stage } s) 
        =\lim_{s\rightarrow\infty} \prod_{\text{even }n\leq r \leq s} ( 1-p_r) 
    \end{split}
    \end{equation}
which converges to a non-zero value if and only if the following series converges \cite{Jeffreys:2006}, \begin{equation}\label{eq_series} \lim_{s\rightarrow\infty}\sum_{\text{even }n\leq r \leq s}^{\infty}p_r. \end{equation} 
Rearranging equation \eqref{pr_def} we have,
    \begin{equation}\begin{split}
        p_r=\frac{\sum_{l=1}^{r}\binom{r-1}{l-1}t_l}{\lambda(r,0)}=\frac{1}{r}\bigg( \frac{\sum_{l=1}^{r}\frac{r!}{(r-l)!(l-1)!}t_l}{\lambda(r,0)}\bigg)\geq \frac{1}{r}\bigg(\frac{\sum_{l=1}^{r}\binom{r}{l}t_l}{t_0+\sum_{l=1}^{r}\binom{r}{l}t_l}\bigg) \geq \frac{1}{r}\bigg(\frac{t_k}{t_0+t_k}\bigg)
        \end{split}
    \end{equation}
    and therefore the series \eqref{eq_series} is divergent and probability \eqref{eq_prob_020321} vanishes.
    
The argument can be adapted to show that every element has an element to its past by letting $x$ be a minimal element and letting $r$ take odd values. \end{proof}

\paragraph{Observables:} Identifying the observables of quantum gravity is a challenge shared by all approaches. Within the sequential growth paradigm, the candidate observables are the measurable events that are covariant: a  measurable event $\mathcal{E}$ is covariant if $\lc{C}\in\mathcal{E}\implies\lc{C}'\in\mathcal{E}$ whenever $\lc{C}\cong\lc{C}'$.  The challenge is to understand which of these candidate covariant events have a comprehensible physical interpretation.
A rich class of observables known as ``stem-events'' has been identified within the past-finite sequential growth framework \cite{Brightwell:2002yu, Brightwell:2002vw}. Each ``stem-event'' corresponds to a logical combination of statements about which finite orders are stems\footnote{A finite order $S$ is a stem in the order $C$ if there exists a representative of $S$ that is a stem in some representative of $C$. A finite order $S$ is a stem in the labeled causet $\lc{C}$ if the order $S$ is a stem in the order $[\lc{C}]$ \cite{Dowker:2019qiz}.} in the growing causet.

What are the analogous observables within the alternating growth framework? Stem-events are indeed measurable in the alternating framework:

\begin{lemma} Stem-events are measureable in an alternating growth model. \end{lemma}
\begin{proof}
First, we give a precise definition of stem-events within the alternating growth framework. For each finite order $C_n$ define the set,
\begin{equation}\label{stem_set_def}\begin{split}
stem(C_n) := &\{  \lc{D} \in \tilde{\Omega}_{\mathbb{Z}} \ |\  C_n \ \textrm{is a stem in} \ \lc{D} \}.
\end{split}
\end{equation}
A stem-event is an element of the $\sigma$-algebra generated by the collection of the $stem(C_n)$'s. Now, we show that each $stem(C_n)$ can be constructed countably from the cylinder sets associated with the nodes of alternating poscau and the result follows.

Let $\lc{C}_n$ be an arbitrary representative of $C_n$. Let $\lc{\Omega}_{m,k}(C_n)$ denote the set of causets $\lc{D}_k$ of cardinality $k$ which are nodes in alternating poscau satifying the following:  $\lc{D}_k$ contains a stem that does not contain any element outside of the interval $[-m,m]$ and that is isomorphic to $\lc{C}_n$. Take the union over the set $\lc{\Omega}_{m,k}(C_n)$ of the associated cylinder sets and call that union $\Gamma_{m,k}(C_n)$:

\begin{equation}\Gamma_{m,k}(C_n):= \bigcup_{\lc{D}_k\in \lc{\Omega}_{m,k}(C_n)} cyl(\lc{D}_k). \end{equation}

Then take the intersection over $k$,
\begin{equation}\Gamma_{m}(C_n):= \bigcap_k \Gamma_{m,k}(C_n),\end{equation}

and the union over $m$,
\begin{equation}stem(C_n)= \bigcup_m \Gamma_{m}(C_n).\end{equation}

\end{proof}
 But the freedom to grow past-infinite causal sets means that the stem-events have a weak distinguishing power---they tell us nothing about the past-infinite part of a casual set and they cannot distinguish between causets with no minimal elements which have no stems. We can make progress by noticing that stems are to past-finite growth what convex-suborders are to alternating growth. The ordering of labeled poscau is determined by the stem relation (\textit{i.e.} the order of labeled poscau is order-by-inclusion-as-stem, cf. definition \ref{def_lab_poscau}), while the ordering of alternating poscau is order-by-inclusion-as-convex-subcauset (cf. definition \ref{def_alt_poscau}). Each node in labeled poscau is a stem in the growing causet, while each node in alternating poscau is a convex-subcauset in the growing causet. Therefore, we propose that ``convex-events'' are the observables for alternating growth, as stem-events are for past-finite growth. 

To make this precise, for each finite order $C_n$ let $convex(C_n)\subset\lc{\Omega}_{\mathbb{Z}}$ be the set of causets that contain $C_n$ as a convex-suborder.

First we prove:
\begin{lemma}
For each finite order $C_n$, $convex(C_n)$ is measureable in an alternating growth model. 
\end{lemma}
\begin{proof}
$C_n$ is a convex-suborder in causet $\lc{C} \in \lc{\Omega}_{\mathbb{Z}}$ 
if and only if there exists a finite integer $N$ such that 
$C_n$ is a convex-suborder in the partial causet $\lc{C}|_{[-N, N]} $ which is $\lc{C}$ restricted to the interval $[-N,N]$.

For each $N\in \mathbb{N}$, let  $\Gamma_N (C_n)  := \bigcup_{\lc{D}_N } cyl(\lc{D}_N) $, where the union is over all labeled causets of cardinality N which are nodes in alternating poscau, $\lc{D}_N$, such that $C_n$ is a convex-suborder of $\lc{D}_N$. 

Then we have
\begin{align}
convex(C_n) = \bigcup_N \Gamma_N(C_n) \,.
\end{align}
\end{proof}

By definition, $convex(C_n)$ is a covariant event and is therefore in the physical $\sigma$-algebra of covariant measureable events. 

Let us also define generally a ``convex-event'' to be any event in the $\sigma$-algebra generated by all the  $convex(C_n)$'s.  Each convex-event is then a covariant measurable event with a clear physical meaning---it corresponds to a logical combination of statements about which finite orders are convex-suborders in the growing causet.

Convex-events form a large class of observables which provide us with information about the structure of the causal set. But they cannot distinguish between pairs of ``convex-rogues'', pairs of non order-isomorphic causal sets that have the same convex-suborders (an example is shown in figure \ref{convex_rogues}). In the past-finite framework, the stem-events are also not fully-distinguishing since they fail to distinguish between pairs of ``rogues''\footnote{If a pair of non order-isomorphic causal sets, $\lc{C},\lc{D}\in\lc{\Omega}_{\mathbb{N}}$, have the same stems as each other then each is called a ``rogue'' and together they form a ``rogue pair''. If $\lc{C}$ and $\lc{D}$ are a rogue pair then every stem-event contains either both or neither.}. However it was shown in \cite{Brightwell:2002vw} that in any CSG dynamics the set of rogues has measure zero and therefore, in a precise sense, the stem-events exhaust the set of observables in any CSG dynamics. Crucially, the result of \cite{Brightwell:2002vw} depends on the specifics of the CSG dynamics and does not hold for every random walk on labeled poscau but only for those models in which the set of rogues has measure zero.

Investigating the consequences of the claim that convex-events exhaust the comprehensible observables in an alternating CSG dynamics, we find that in the only Alternating CSG that satisfies DGC --- namely  Alternating Transitive Percolation --- the convex-events fail to provide any useful predictions. This is because 
in ATP (and TP)  \textit{every} finite order is almost surely a convex-suborder in the  causet grown: \textit{i.e.} the measure of every event $convex(C_n)$ is equal to 1 \cite{Brightwell:2016}.  If we define a model to be ``deterministic with respect to convex-events'' if every convex-event has measure zero or one, then ATP is deterministic with respect to the convex-events.
Indeed, the causet grown will almost surely contain infinitely many copies of every convex-suborder: no matter where  you are in an ATP universe, a copy of each finite order will occur in your future if you wait long enough just as a given finite bit string will almost surely occur infinitely many times in an infinite random string. The stem-events in TP, anchored as they are to the beginning, do not suffer from this problem.
So convex-events cannot, with probability one, distinguish between any two infinite causets grown in ATP -- any two infinite causets grown in ATP are almost surely a convex-rogue pair -- and one can make no useful predictions using convex-events. 

This example of ATP is important because TP holds a special position in the class of CSG models. It is a fixed point under the transformations known as cosmic renormalisation \cite{Martin:2000js} that are the basis for the hope that causal set cosmology might be self-tuning and avoid the fine-tuning of cosmological parameters we find in our current standard cosmological model \cite{Sorkin:1998hi}.  This failure  means one must give up on ATP as a useful cosmological model or one must look harder for meaningful observables or one gives up on growth models that allow two-way infinite causets altogether. 

In the rest of the paper, we take the first choice above: we adhere to the proposal of convex-events as the meaningful observables, accept that this means that ATP is not a useful cosmological model and explore models that allow two-way infinite causets  in which there are non-trivial predictions about convex-events. 
 First, we show that not every Alternating CSG model is deterministic with respect to the convex-events:
\begin{claim}\label{claim090702} 
An Alternating CSG dynamics  is not deterministic with respect to convex-events if its couplings are given by, \begin{equation}\label{eq300501} t_0=1 \text{ and } t_n=f(n)\lambda(n-1,0) \ \forall n\geq 1,\end{equation} where $f(n)$ is a function satisfying $\sum_1^\infty \frac{1}{f(n)}<\infty$ (e.g. $f(n)=x^n$ with $x>1$ or $f(n)=n^s$ with $s>1$). \end{claim}

\begin{proof}
Let $A_2$ denote the 2-antichain order, and let $C_{\infty}$ denote the two-way infinite chain order. Note that $\mathbb{P}(C_{\infty})=1-\mathbb{P}(convex(A_2))$, where  $\mathbb{P}(C_{\infty})$ is the probability of growing $C_{\infty}$ and $\mathbb{P}(convex(A_2))$ is the measure of $convex(A_2)$. By considering stage 1 of the growth we see that $\mathbb{P}(convex(A_2))>t_0/\lambda(1,0)>0$ in any alternating CSG dynamics. We will show that in the dynamics (\ref{eq300501}), $\mathbb{P}(C_{\infty})>0$  and therefore $0<\mathbb{P}(convex(A_2))<1$ and the result follows.

Now,  $\mathbb{P}(C_{\infty})=\prod_{n>0} p_n$, 
where (as in claim \ref{immortality_claim}) $p_n$ is the effective parameter given by,
\begin{equation}\label{eff_par_eq}\begin{split}p_n=\frac{\sum_{k=0}^{n-1}\binom{n-1}{k}t_{k+1}}{\lambda(n,0)}=&\frac{\lambda(n,0)-\lambda(n-1,0)}{\lambda(n,0)},\end{split}\end{equation}
and the product converges to a non-zero value if and only if the series $\sum ( 1-p_n)$ converges \cite{Jeffreys:2006}. We can write the $m^{th}$ term of this series as, \begin{equation}\label{eq17052}\begin{split}1-p_m =\frac{\lambda(m-1,0)}{\lambda(m,0)}=\bigg(\frac{\sum_{r=0}^{m-1}\binom{m}{r}t_r}{\lambda(m-1,0)}+\frac{t_m}{\lambda(m-1,0)}\bigg)^{-1},\end{split}\end{equation}
and then substitute the couplings given in \eqref{eq300501} to find,
\begin{equation}\begin{split} 1-p_m=\bigg(\frac{\sum_{r=0}^{m-1}\binom{m}{r}t_r}{\lambda(m-1,0)}+f(m)\bigg)^{-1}\leq \frac{1}{f(m)}.\end{split}\end{equation}

It follows that in the models given in \eqref{eq300501} the sum $\sum (1-p_n)$ converges by the comparison test against $\sum \frac{1}{f(n)}$ and hence $\mathbb{P}(C_{\infty})>0$.
\end{proof}

The existence of non-deterministic alternating growth models encourages us to continue to explore dynamics that allow two-way infinite causets to grow. We might  use the concept of convex-events to formulate constraints or guiding principles in searching for interesting alternating growth dynamics -- e.g. a stronger condition than that the dynamics is not deterministic w.r.t. convex-events is that the dynamics almost surely does not generate convex-rogues. It is not known whether the dynamics \eqref{eq300501}
satisfies this condition. 

\begin{figure}[h]
  \centering
	\includegraphics[width=0.4\textwidth]{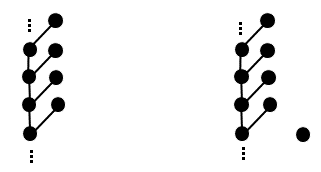}
	\caption{The ``infinite comb'' (left) and the infinite comb disjoint union a single element (right) are convex-rogues since they contain the same convex-suborders as each other.}
	\label{convex_rogues}
\end{figure}

\vspace{2mm}
In summary, in this section we generalised the sequential growth paradigm to accommodate two-way infinite cosmologies. The resulting alternating framework generates causets that have a natural labeling by $\mathbb{Z}$. We considered what form various physical conditions take in the alternating framework and whether an alternating generalisation of the CSG models satisfy them. Finally, we identified the convex-events as a physical class of observables. The convex-events cannot discriminate between causets grown in an ATP model, which model is the only alternating CSG model that satisfies DGC. This means that if we want to demand DGC in an alternating growth model, that model cannot be an alternating CSG model. 

In the next section we use the convex-events to provide an alternative to alternating sequential growth: a covariant framework for two-way infinite growth, analogous to the existing covariant framework for the growth of past finite causets \cite{Dowker:2019qiz, Zalel:2020oyf}. 
\FloatBarrier
\pagebreak

\section{Covariant growth}\label{section_cov_growth}

Sequential growth models are named for the way they are defined, with the causal set elements being born in a sequence of stages, with specified transition probabilities for the possible transitions at each stage. This linear order of the stages is a gauge---a kind of supertime---since it is a tenet of Causal Set Theory that only the partial order of the causet itself is physical. In other words, the definition of sequential growth models makes the elements of the growing causal set mathematically distinguishable or ``labeled'' --- since elements are distinguished/labeled by the stage at which they are born --- but some of this labeling information is unphysical since in Causal Set Theory only the order-isomorphism class of the causet is physical. The dissonance between the labeled nature of sequential growth and the label-independent nature of the physical world finds a resolution once one has identified the covariant, label-independent observables and restricted oneself to making statements only about them. Thus, sequential growth models are a proof of concept for the growth dynamics paradigm and a playground in which to explore the dichotomy of being and becoming  \cite{Sorkin:2007qh, spacetimeatoms}.

Covariant growth of past-finite causets is an alternative framework to sequential growth in which label independence is manifest from the start \cite{Dowker:2019qiz,Zalel:2020oyf}. Its motivation is rooted in the notion of partially ordered growth or \textit{asynchronous becoming}, in which the world comes into being --- becomes ---  in a manner compatible with a lack of physical global time through a partially ordered process of the birth of spacetime atoms  \cite{Sorkin:2007qh, spacetimeatoms}. Covariant growth models seek to bypass the introduction of the unphysical gauge in sequential growth --- the linear order of the stages at which the causet elements are born one by one --- and to deal only with covariant quantities throughout. This is an ambitious project and we anticipate that the struggle between the local nature of the dynamics of a gauge field and the global nature of gauge invariant quantities will play out in pursuing it. 

 Thus far, covariant growth has only been explored in the context of past-finite orders where the dynamics takes the form of a random walk up \textit{covtree}, a partial order that is a directed tree whose nodes are sets of orders.\footnote{Recall that ``order'' is short for ``order-isomorphism class'' (see section \ref{subsec_orders}).} At level $n$ of covtree, each node is a set of $n$-orders, interpreted as the set of $n$-stems of the growing past-finite causet. This interpretation is founded on the theorem that for each inextendible path up covtree there indeed exists an infinite order whose $n$-stems form the node in that path at level $n$ \cite{Dowker:2019qiz}. This dynamics pertains to covariant properties of the causet from the outset and no labeling is introduced. The random walk progresses in stages from each level of covtree to the next. At the beginning of stage $n$ of the random walk, the ``state'' of the process is a node in level $n-1$ -- that is, the $(n-1)$-stems of the growing causet have already been chosen. At stage $n$, the walk transitions to a new node in level $(n-1)$:  i.e. the set of $n$-stems of the growing causal set universe is chosen at random according to the transition probabilities of the model. And so on. 
 
  Note that in this scenario of covariant growth, the manifest label-independence comes at the ``cost'' of the model not making direct reference to the process of the birth of \textit{individual} spacetime atoms: in  a  sequential growth model -- i.e. a random walk up labeled poscau -- element $n$ is born at stage $n$ and that is not the case in covariant growth on covtree. A covtree node at stage $n$,  $\Gamma_n$,  is a collection of $n$-orders and corresponds to the statement ``$\Gamma_n$ is the set of $n$-stems in the growing universe''. In a real sense, however, in moving from a poscau process to a covtree process one is losing what one does not actually have.  Since, in a walk on labeled poscau, tempting as it is to interpret the node at stage $n$
as representing a momentary state of a growing order this is an unphysical picture because the concept of stage $n$ has no physical meaning:  there is no ``God's eye view'' of the universe in asynchronous becoming \cite{Sorkin:2007qh}.  Here we see the struggle between locality and global-ness inherent in a gauge theory.   

Our aim is to create a covariant framework  for two-way infinite growth and construct the analogue of covtree. The construction of covtree was motivated by the fact that the stem-events exhaust the set of observables in CSG models \cite{Brightwell:2002vw}. Indeed, covtree's algebra of observables is equal to the algebra of stem-events \cite{Dowker:2019qiz}. Therefore, pursuing further the analogy between stems and convex-suborders, in the rest of the paper we introduce and explore a new covariant framework, which we call $\Z$-covtree, whose sample space is $\Omega_{\Z}$ and whose set of observables is exactly the set of convex-events. We will see that the structure of 
$\Z$-covtree is very different from covtree. We will construct $\Z$-covtree via an intermediate construction of a larger tree we call convex-covtree. The next subsection is devoted to defining convex-covtree. 


\subsection{Defining convex-covtree}
\label{sec:CovtreePastInfinite}

Recall that, for any positive integer $n$, the set of $n$-orders is called $\Omega(n)$. Let $\Gamma_n$ denote a subset of $\Omega(n)$. Recall also that an $n$-convex-suborder means ``a convex-suborder of cardinality $n$''. Convex-covtree is a partial order, a directed tree whose nodes at level $n$ are subsets of $\Omega(n)$: a subset $\Gamma_n\subset\Omega(n)$ is a node in convex-covtree if and only if it is the set of $n$-convex-suborders of some (finite or infinite) order $C$. In the following, we formalise the definition of  convex-covtree. 

\begin{definition}\label{cert_def} An order $C$ is a \textbf{certificate} of $\Gamma_n$ if $\Gamma_n$ is the set of  $n$-convex-suborders of $C$.  A \textbf{labeled certificate} of $\Gamma_n$ is a representative of a certificate of $\Gamma_n$.
\end{definition}

A certificate may be finite or infinite, and if it is infinite it may be past-finite, future-finite or two-way infinite. Note that some $\Gamma_n\subset\Omega(n)$ have no certificates at all. If $\Gamma_n$ has an infinite certificate then it has a finite certificate, but the converse is not true. Examples are shown in figure \ref{fig_cert}.\footnote{Note that definition \ref{cert_def} of \textit{certificate} is different to that in \cite{Dowker:2019qiz} where a certificate of $\Gamma_n$ is an order whose set of $n$-\textit{stems} is $\Gamma_n$. If $C$ is a certificate of $\Gamma_n$ by definition \ref{cert_def}, then $C$ certifies that $\Gamma_n$ is a node in convex-covtree. If $C$ is a certificate of $\Gamma_n$ by the definition in \cite{Dowker:2019qiz}, then $C$ certifies that $\Gamma_n$ is a node in covtree. The properties of the certificates depend on which definition of certificate is used, \textit{e.g.} using the definition in \cite{Dowker:2019qiz} $\Gamma_n$ has an infinite certificate if and only if it has a finite certificate, while using definition \ref{cert_def} the existence of a finite certificate does not imply the existence of an infinite certifcate.}

\begin{figure}
    \centering
    \includegraphics[scale=0.5]{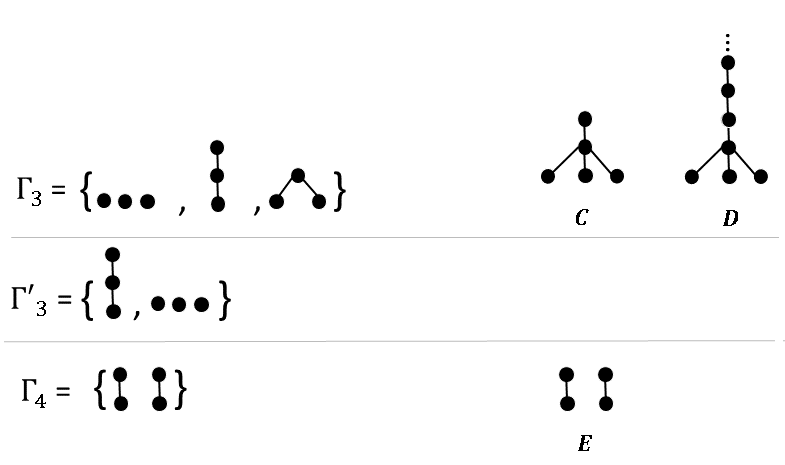}
    \caption{Illustration of certificates. $C$ and $D$ are certificates of $\Gamma_3$. $\Gamma'_3$ has no certificates since any order that contains the 3-chain and the 3-antichain as 3-convex-suborders also contains the ``L'' order as a 3-convex-suborder. $E$ is a certificate of $\Gamma_4$. $\Gamma_4$ has no infinite certificates.}
    \label{fig_cert}
\end{figure}

We use $\chi$ to denote the set of all $\Gamma_n$'s, for all $n$, that have at least one certificate:
    \vspace{-4mm}
    \begin{equation}
        \chi:=\bigcup_{n\in\N}~\{\Gamma_n\subseteq\Omega(n)~|~\exists\text{ a certificate of }\Gamma_n\}.
    \end{equation}
$\chi$ is the ground-set of convex-covtree. To define the partial order on $\chi$, we introduce the map  ${{\mathcal O}}_c^{-}$:
\begin{definition}
    For any $n$ and any set $\Gamma_n$ of $n$-orders, the map ${{\mathcal O}}_c^{-}$  takes $\Gamma_n$ to the set of $(n-1)$-convex-suborders of elements of $\Gamma_n$:
\vspace{-4mm}
\begin{equation} 
{\mathcal O}_c^{-}(\Gamma_n):=\{B \in \Omega(n-1)\ | \  \exists \ A\in \Gamma_n\ \mathrm{ s.t. }\ B \text{ is an $(n-1)$-convex-suborder in } A \}
\,. \label{o_minus_defn}
\end{equation}
\end{definition}

One way to think about the operation ${{\mathcal O}}_c^{-}$ on $\Gamma_n$ is to pick an $n$-order in $\Gamma_n$ and delete a maximal or minimal element of it to form an $(n-1)$-order. The set
${{\mathcal O}}_c^{-}(\Gamma_n)$ is the set of all $(n-1)$-orders that can be formed in this way. An illustration is shown in figure \ref{fig_operator}.

\begin{figure}[htpb]
    \centering
    \includegraphics[scale=0.5]{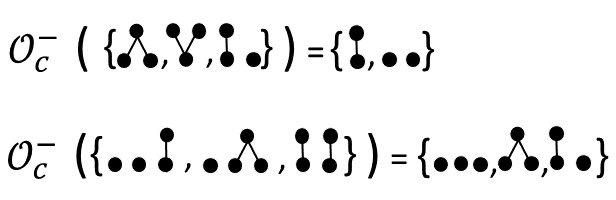}
    \caption{Illustration of the operation  ${{\mathcal O}}_c^{-}$.}
    \label{fig_operator}
\end{figure}

\begin{lemma} If $\Gamma_n \in \chi$ then ${{\mathcal O}}_c^{-}(\Gamma_n) \in \chi$. 
\end{lemma}
\begin{proof}
There exists a certificate $C$ of $\Gamma_n$. Each element of $\Gamma_n$ is a convex-suborder of $C$. 
So each convex-suborder of each element of $\Gamma_n$ is a convex-suborder of $C$. An 
 $(n-1)$-order  is an $(n-1)$-convex-suborder of $C$ if and only if it is a convex-suborder of some $n$-convex-suborder of $C$. Therefore $C$ is a certificate of ${{\mathcal O}}_c^{-}(\Gamma_n) $. 
 \end{proof}

We can now give the definition of convex-covtree:
\begin{definition}
    Convex-covtree is the partial order $(\chi,\prec)$, where $\Gamma_n\in\chi$ is directly above -- covers -- $\mathcal{O}_c^{-}(\Gamma_n)\in\chi$.
\end{definition}

The nodes in the first three levels of convex-covtree are shown in figures \ref{fig:covtree_pastinfinite} and \ref{fig:covtree_pastinfinite2}.
The construction of convex-covtree is closely analogous to the construction of covtree in \cite{Dowker:2019qiz}, with the concept of convex-suborder replacing the concept of stem. Indeed, the two resulting structures share some features, including: (1) if $C$ is a certificate of a node $\Gamma_n$ then $C$ is a certificate of all nodes below $\Gamma_n$ and (2) every inextendible path has a certificate (as we will prove for convex-covtree in lemma \ref{lemma230221} and proposition \ref{lemma_1_211020} below), where the certificate of a path is defined as,
\begin{definition}\label{def_240221} An order $C$ is a \textbf{certificate of a path} $\mathcal{P}$ if it is a certificate of every node in $\mathcal{P}$.
\end{definition}
\noindent Properties (1) and (2) allow us to interpret a random walk up convex-covtree as a covariant process of growth: the growing order is a certificate of the path that is taken by the random walk. Each node in the path corresponds to a covariant property of the growing order, \textit{i.e.} $\Gamma_n$ is the set of $n$-convex-suborders of the growing order. At stage $n$, the walk transitions from the set of $(n-1)$-convex-suborders of the growing order to the set of $n$-convex-suborders.  At each stage of the random process, more physical information about the growing order is acquired.  

\subsection{Sample space for convex-covtree} 
In labeled alternating sequential growth models, there is a $1-1$ correspondence between the set of paths on alternating poscau and the set of labeled causets, ${\lc{\Omega}}_{\Z}$, and we refer to the latter as the sample space of the process. Events in the event algebra are subsets of this sample space. Covariant events are covariant subsets of this sample space. 

The framework of random walks up convex-covtree, is motivated by doing away with mention of labeled causets from the very start. In keeping with this, but keeping to the physical interpretation that the process is producing a growing order,   we conceive of the sample space of the process, not as a set of labeled causets,  but as a set of orders. 

\begin{definition} The  \textit{sample space} of a random walk on convex-covtree is the set of orders that are certificates of inextendible (maximal) paths in convex-covtree. 
\end{definition}
There is no $1-1$ correspondence between inextendible  paths in convex-covtree and orders: we have already seen this in the example of ATP where almost surely any causal set grown in ATP has the same convex-suborders as any other. So the single path in convex-covtree containing the node at level $n$ that is the set of all $n$-orders, for all $n$, has all the ATP orders as certificates. We will come back to this point in the discussion.  

Now, we can ask: which orders are in this sample space for walks on convex-covtree? In contrast to all  growth models defined to date,  a walk up convex-covtree can produce \textit{finite} orders. This is because convex-covtree contains maximal nodes, so some of its inextendible paths are finite. A finite inextendible path has one unique finite certificate, and so if a random walk ends at a maximal element of convex-covtree,  then a finite order is generated and the universe has a beginning \textit{and} an end.  This result and others about maximal nodes and finite inextendible paths will be proved in the next subsection \ref{subsec_inext}. The certificates of infinite paths are necessarily infinite (since they contain $n$-convex-suborders for every $n>0$) and every infinite order (past-finite, future-finite or neither) is a certificate of some infinite path.

In summary, the sample space of a random walk on convex-covtree contains all infinite orders and many (but not all) finite orders. 
It is natural to ask whether there is a way to consistently restrict the sample space to $\Omega_{\Z}$, in order that the sample space matches that of the alternating sequential growth models of the previous section. We will show in section \ref{subsec_inf_paths} that  this can be done and that in this case the observables are the convex-events. We will also show that an inconsistency arises ((2) is violated) when restricting the sample space to $\Omega_{\mathbb{N}}$, suggesting that convex-suborders are unsuitable for describing past-finite growth.

\begin{figure}[h]
    \centering
    \includegraphics[scale=0.6]{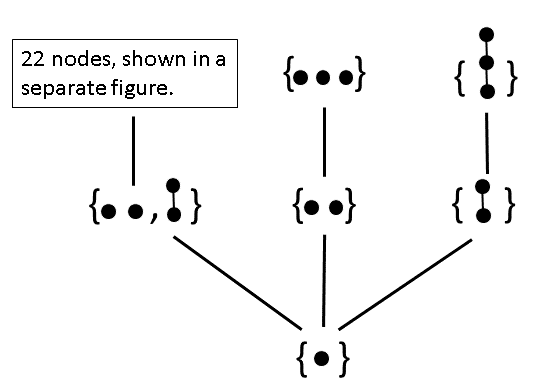}
    \caption{The first three levels of convex-covtree.}
    \label{fig:covtree_pastinfinite}
\end{figure}

\begin{figure}[h]
    \centering
    \includegraphics[scale=0.3]{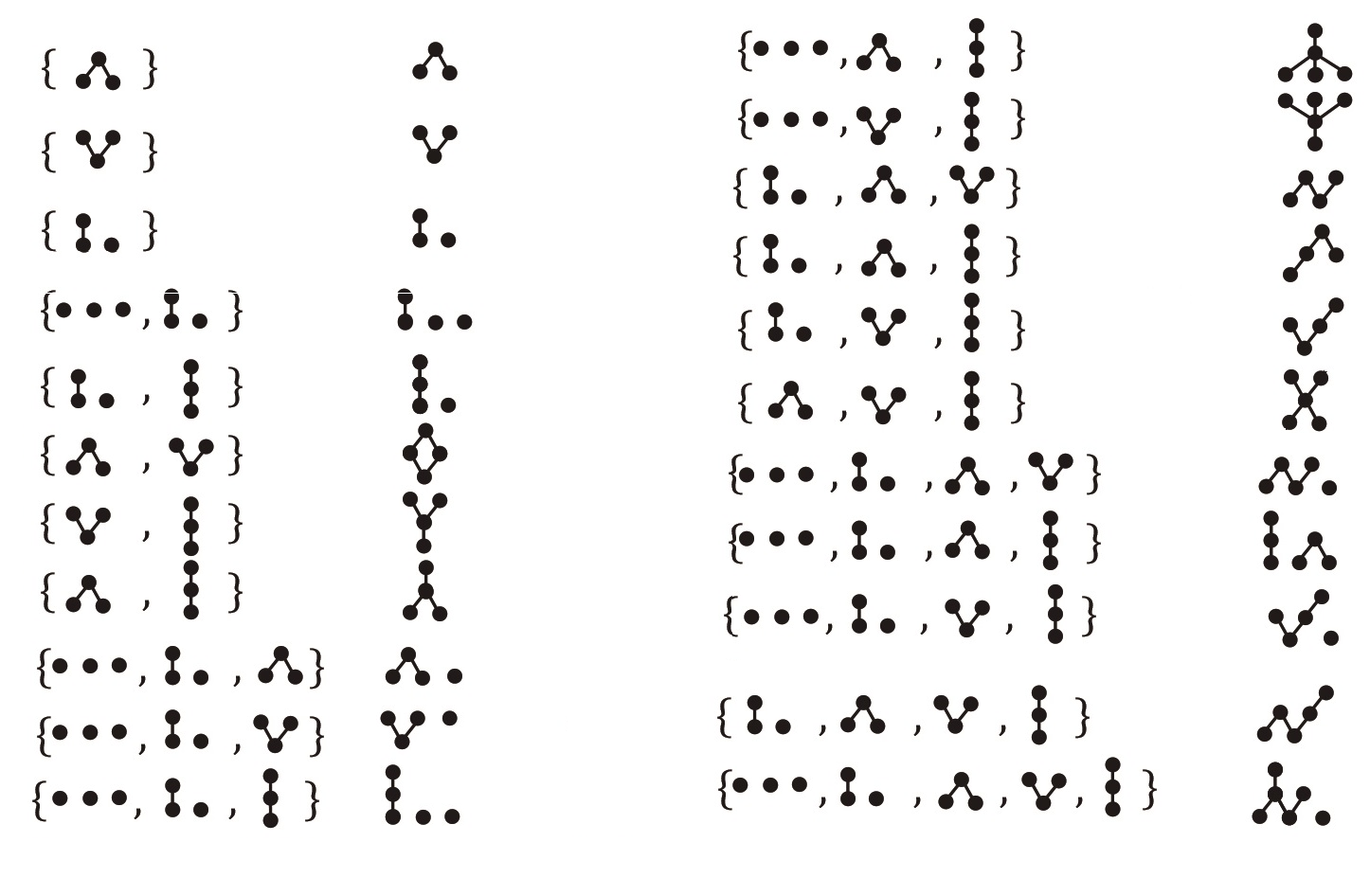}
    \caption{22 nodes of convex-covtree and their certificates. These are the level 3 nodes that appear directly above the doublet.}
    \label{fig:covtree_pastinfinite2}
\end{figure}

\FloatBarrier
\subsection{Finite inextendible paths} \label{subsec_inext}
 
A \textit{maximal node} is a node that has no descendants---it is maximal in the convex-covtree partial order. A subset $\Gamma_n\subset\Omega(n)$ is a \textit{singleton} if it contains only a single $n$-order, \textit{i.e.} $\Gamma_n=\{C\}$. Note that every singleton is a node in convex-covtree since if $\Gamma_n=\{C\}$ then $C$ is a certificate of $\Gamma_n$. 

\begin{lemma}\label{lemma_max_sing} A maximal node is a singleton $\Gamma_n=\{C\}$ whose only certificate is $C$. \end{lemma}

\begin{proof}
Let $\Gamma_n=\{C\}$ and let its only certificate be $C$. Suppose for contradiction that $\Gamma_{n+1}\succ\Gamma_n$. Then there exists some $D$ with cardinality $>n$ that is a certificate of $\Gamma_{n+1}$ and hence of $\Gamma_{n}$. Contradition. Therefore $\Gamma_n$ is maximal.

Suppose that $\Gamma_n=\{C\}$ has a certificate $D\not=C$. Then $D$ has cardinality $>n$ and therefore $\{D\}\succ\Gamma_n\implies$ $\Gamma_n$ is not maximal. Similarly, if $\Gamma_n$ is not a singleton then it has a certificate $D$ with cardinality $>n\implies \{D\}\succ\Gamma_n$.
\end{proof}

\noindent The singleton $\Gamma_4=\{$\twoch \twoch$ \}$ (also shown in figure \ref{fig_cert}) is an example of a maximal node. To see that $\Gamma_4$ has no certificate of cardinality $>4$ it is sufficient to attempt to construct such a certificate by adding a single element to (a representative of) \twoch \twoch . For example, we can add the new element to form the $5$-order \twoch \twoch \oneach , but this $5$-order is not a certificate of $\Gamma_4$ since it contains the \twoch \twoach \ as a $4$-convex-suborder. Continuing in this way, we find that it is impossible to form a certificate of $\Gamma_4$ by adding an element to \twoch \twoch. Indeed, \twoch \twoch  is the unique certificate of $\Gamma_4$.

The existence of maximal nodes implies the existence of finite inextendible paths. We can characterise finite inextendible paths as follows:

\begin{proposition}\label{lemma_1902} An inextendible path $\mathcal{P}$ is finite if and only if it contains a singleton $\{C_n\}$, where $C_n$ is not the $n$-chain or the $n$-antichain. \end{proposition}

\noindent To prove proposition \ref{lemma_1902} we will need:
\begin{lemma} Let $C_n$ be an $n$-order that is not the $n$-chain or the $n$-antichain. Then every certifcate of $\{C_n\}$ has cardinality less than $n^2$.
\end{lemma}

\begin{proof} For any (finite or infinite) order $C$, let $w(C)$ and $h(C)$ denote the width and height of $C$, respectively. Note that $|C|\leq h(C)w(C)$. Additionally, if $C$ is a certificate of $\{C_n\}$ then $w(C)=w(C_n)<n$. We will show that if $C$ is a certificate of $\{C_n\}$ then $h(C)\leq n$ and the result follows.

Let $C$ be an order with $h(C)>n$ and suppose for contradiction that $C$ is certificate of $\{C_n\}$. Let $D$ be a chain of length $n+1$ in $C$ and let $\mathcal{H}$ be the convex hull of $D$. Then $|\mathcal{H}|=n+k$ for some $k>0$. Note that $\mathcal{H}$ is an interval by construction, \textit{i.e.} it has a single maximal element and a single minimal element. We will now show by induction that $\mathcal{H}$ is a chain and therefore $C$ is not a certificate of $\{C_n\}$.

One way to obtain $C_n$ from $\mathcal{H}$ is to remove the minimal element of $\mathcal{H}$ to form the order $\mathcal{H}_{-1}$, then remove a minimal element of $\mathcal{H}_{-1}$ to form $\mathcal{H}_{-2}$ and so on until $\mathcal{H}_{-k}=C_n$. Since $\mathcal{H}$ has a unique maximal element, $\mathcal{H}_{-k}=C_n$ has a unique maximal element.

Another way to obtain $C_n$ from $\mathcal{H}$ is to remove the maximal element of $\mathcal{H}$ to form the order $\mathcal{H}^{-1}$, then remove a minimal element of $\mathcal{H}^{-1}$ to form $\mathcal{H}^{-1}_{-1}$, then remove a minimal element of $\mathcal{H}^{-1}_{-1}$ to form $\mathcal{H}^{-1}_{-2}$ and continue to remove minimal elements until $\mathcal{H}^{-1}_{-k+1}=C_n$. The top level of $\mathcal{H}^{-1}_{-k+1}=C_n$ is level $h(C)-1$ of $\mathcal{H}$, and since $C_n$ has a unique maximal element we learn that $\mathcal{H}$ has only one element at level $h(C)-1$.

Suppose $\mathcal{H}$ has only one element at each of the levels $h(C), h(C)-1, ..., h(C)-r+1$ for some $r<h(C)$. Then $\mathcal{H}^{-r}_{-k+r}=C_n$ is constructed by removing the top $r$ levels of $\mathcal{H}$ and therefore the top level of $\mathcal{H}^{-r}_{-k+r}=C_n$ is level $h(C)-r$ of $\mathcal{H}$. Since $\mathcal{H}^{-r}_{-k+r}=C_n$ has a unique maximal element we learn that $\mathcal{H}$ has only one element at level $h(C)-r$. Therefore, by induction $\mathcal{H}$ has a single element at each level, \textit{i.e.} $\mathcal{H}$ is a chain.
\end{proof}

\begin{proof}[Proof to proposition \ref{lemma_1902}]
Let $\{C_n\}\in\mathcal{P}$ and suppose for contradiction that $\mathcal{P}$ is infinite. Then for any $N>n^2$ there exists a node $\Gamma_N\in\mathcal{P}$. Let $C$ denote a certificate of $\Gamma_N$ and note that $|C|\geq N> n^2$. Since $\Gamma_N\succ \{C_n\}$, $C$ is a certificate of $\{C_n\}$. Contradiction. That the converse is true follows from the fact that every maximal node is a singleton (lemma \ref{lemma_max_sing}).
\end{proof}

We can also identify the certificates of the finite inextendible paths:
\begin{lemma}\label{lemma230221} If $\mathcal{P}=\Gamma_1\prec\Gamma_2\prec..\prec\Gamma_k$ is a finite inextendible path then $C_k\in\Gamma_k$ is the unique certificate of $\mathcal{P}$.
\end{lemma}

\begin{proof} Clearly, $C_k$ is a certificate of $\mathcal{P}$ and there are no other certificates of $\mathcal{P}$ with cardinality $\leq k$. Suppose $C_l$ is a certificate of $\mathcal{P}$ with cardinality $l>k$. Then $\{C_l\}\succ\Gamma_k$. Contradiction.
\end{proof}

\noindent A corollary is that the corresponding sample space contains spacetimes of finite volume, namely the certificates of the finite inextendible paths. An $n$-order $C_n$ is an element of the sample space if there is no order $D\not=C_n$ whose only $n$-convex-suborder is $C_n$. For example, the sample space contains the $4$-order \twoch \twoch , but it does not contain the ``L'' order,  \twoch \oneach , since $\{  \twoch \oneach \}\prec \{\twoch \twoch \}$.

\begin{lemma} The sample space contains countably many finite orders. \end{lemma}
\begin{proof} Let $Q(n)$ denote the number of singletons $\{C_n\}$ at level $n$ in convex-covtree, where $C_n$ is not the $n$-chain or the $n$-antichain. Each of these $Q(n)$ nodes is in at least one finite path and no two are in the same path. Therefore there are at least $\lim_{n\rightarrow\infty}Q(n)$ finite inextendible paths. \end{proof}

It may seem that the sample space is entropically dominated by the infinite orders, as there are uncountably many of these and only countably many finite orders. But if one assigns transition probabilities uniformly such that the probabilities to transition from a given node of convex-covtree to any of its children are equal, then the event that spacetime has finite cardinality happens with probability $>\frac{1}{22}$ (since this is the probability of reaching a singleton that does not contain a chain or an antichain by level 3). By Proposition   \ref{lemma_1902} the models  which almost surely produce infinite universes are exactly those that satisfy $\mathbb{P}(\Gamma)=0$ whenever $\Gamma$ is a singleton node that does not contain a chain or an antichain.\footnote{For any $n>1$, if $\Gamma_n$ is a singleton that contains a chain then it is contained in a unique inextendible path,  $\{\ \oneach \}\prec \{\twoch\}\prec\{\threech\}\prec ... \ .$ Similarly, if $\Gamma_n$ is a singleton that contains an antichain  then it is contained in a unique inextendible path,
$\{\oneach \}\prec \{\twoach\}\prec\{\threeach \ \ \ \ \ \}\prec \dots \ .$}

\subsection{Infinite paths and $\Z$-covtree}\label{subsec_inf_paths}
We now prove that:
\begin{proposition}\label{lemma_1_211020}
Every infinite path in convex-covtree has a certificate.
\end{proposition}

Together, lemma \ref{lemma230221} and proposition \ref{lemma_1_211020} enable us to interpret a walk on convex-covtree as a process in which an order grows---they guarantee that each realisation of the walk will produce some order. A path has more than one certificate if its certificates are convex-rogues and, in this case, which certificate is the growing order is up for interpretation (\textit{e.g.} we can consider all certificates of a given path to be physically equivalent). 

To prove proposition \ref{lemma_1_211020}, we adapt the algorithm from \cite{Dowker:2019qiz} that generates a certificate for any infinite path $\mathcal{P}$.
We will need the concept of ``minimal certificates'':
\begin{definition}
Given some $\Gamma_n$, we order its finite certificates by inclusion. Let $C_1$ and $C_2$ be two finite certificates of $\Gamma_n$. We say $C_1\leq C_2$ if and only if $C_1$ is a convex-suborder in $C_2$. A \textbf{minimal certificate} of $\Gamma_n$ is minimal in this order.  
\end{definition}
We will also need the following lemma:
\begin{lemma}\label{claim_theorem_1_prep} Let $\mathcal{P}=\Gamma_1\prec \Gamma_2 \prec \Gamma_3\prec \dots$ be an infinite path in convex-covtree. Then for any $\Gamma_n\in\mathcal{P}$ there exists some $\Gamma_m\in\mathcal{P}$ that contains a certificate of $\Gamma_n$.  
\end{lemma}

\begin{proof}
First, note that it follows from the definition of convex-covtree that if $\Gamma_n$ is a singleton and $\Gamma_m\succ\Gamma_n$ then any $C\in\Gamma_m$ is a certificate of $\Gamma_n$. If $\Gamma_n$ is not a singleton, then every minimal certificate $C$ of $\Gamma_n$ satisfies $n< \vert C\vert \leq N$ where $N := n |\Gamma_n|$. Consider $\Gamma_N\in \mathcal{P}$ and let $D$ be a finite certificate of $\Gamma_N$. Since a certificate of a node is a certificate of all the nodes below it, $D$ is a certificate of $\Gamma_n$.  
Now, at least one minimal certificate of $\Gamma_n$ occurs as a convex-suborder in $D$. Choose one, call it $C$, let $m := \vert C \vert$ and consider 
$\Gamma_m\in \mathcal{P}$.  $\Gamma_m$
is the set of all convex-suborders of cardinality $m$ in $D$ and so $C$ is an element of $\Gamma_m$.  
\end{proof}

\begin{proof}[Proof of proposition \ref{lemma_1_211020}] 
Given an infinite path $\mathcal{P}=\Gamma_1\prec \Gamma_2\prec...$, the following inductive algorithm generates an infinite nested sequence of causal sets,
\newline $\lc{C}_{m_1}\subset \lc{C}_{m_2}\subset \lc{C}_{m_3} \subset...$ :
\newline \textit{ Step 1:}
\newline 1.0) Pick some natural number $m_0>0$ and 
consider $\Gamma_{m_0}\in \mathcal{P}$.
\newline 1.1) By  lemma \ref{claim_theorem_1_prep}, there exists some $\Gamma_{m_1}\in \mathcal{P} $
that contains some certificate $C_{m_1}$ of $\Gamma_{m_0}$. Pick a representative $\lc{C}_{m_1}$ of $C_{m_1}$.
\newline 1.2) Go to step 2.
\newline \textit{ Step} $k>1$:
\newline k.1) By lemma \ref{claim_theorem_1_prep}, there exists some $\Gamma_{m_{k}}\in\mathcal{P}$ that contains some certificate $C_{m_k}$ of $\Gamma_{m_{k-1}}$. Pick a representative  $\lc{C}_{m_k}$ of  ${C}_{m_k}$ such that 
$\lc{C}_{m_{k-1}}$ from the previous step is a sub-causet of $\lc{C}_{m_{k}}$.
\newline k.2) Go to step $k+1$.

\vspace{2mm}
By construction, the union $\lc{C}:=\bigcup_{i=1}^{\infty} \lc{C}_{m_i}$ is order-isomorphic to a labeled certificate of $\mathcal{P}$.
If the ground-set of $\lc{C}$ is $\mathbb{Z}$, $\mathbb{N}$ or $\mathbb{Z}^-$ then $\lc{C}$ is a labeled certificate of the path $\mathcal{P}$. If $\lc{C}$ has ground-set $[p,\infty)$ for some $p\in\mathbb{Z}$ then $\lc{C}$ is order-isomorphic to some causet $\lc{D}$ with ground-set $\mathbb{N}$. In this case $\lc{D}$ is a labeled certificate of $\mathcal{P}$. If $\lc{C}$ has ground-set $(-\infty,p]$ for some $p\in\mathbb{Z}$ then $\lc{C}$ is order-isomorphic to some causet $\lc{E}$ with ground-set $\mathbb{Z}^-$. In this case $\lc{E}$ is a labeled certificate of $\mathcal{P}$. 
Since in each case $\mathcal{P}$ has a labeled certificate, every infinite path has a certificate.
\end{proof}

As previously stated, the sample space of convex-covtree contains all infinite orders and countably many (but not all) finite ones. We set out to find a covariant counterpart to alternating sequential growth, and now we see that convex-covtree is not that framework. We now ask: 
can convex-covtree can be truncated into a tree whose sample space equals $\Omega_{\mathbb{Z}}$?

 For a start, we can consider the subtree of convex-covtree that contains only the nodes that have \textit{infinite} certificates or equivalently the subtree of convex-covtree that is the union of all infinite paths. By truncating the finite inextendible paths we remove the finite orders from the sample space and proposition \ref{lemma_1_211020} guarantees that each inextendible path in this truncated covtree has a certificate in $\Omega$. 
 
 However, there is no guarantee that every path has a certificate in $\Omega_{\mathbb{Z}}$. Indeed, there exist infinite paths that only have certificates in $\Omega_{\mathbb{N}}$ and others that only have certificates in $\Omega_{\mathbb{Z}^-}$. Recall that a certificate of a path is a certificate of all its nodes. Therefore if there exists some $\Gamma_n\in \mathcal{P}$ whose infinite certificates are only in $\Omega_\mathbb{N}$ then $\mathcal{P}$ only has certificates in $\Omega_\mathbb{N}$. 
 
 For example, consider the node $\Gamma_3=\{\lambdacauset \ , \threech\}$ whose unique minimal certificate is $\lambdafour \ $. We can construct any certificate of $\Gamma_3$ by starting with its minimal certificate and then adding elements to it. In particular, if $\Gamma_3$ has a certificate in $\Omega_{\mathbb{Z}}$ or $\Omega_{\mathbb{Z}^-}$ then we should be able to grow a certificate of $\Gamma_3$ by adding an element that is spacelike or to the past of every element in $\lambdafour$. There are 5 ways to add such an element, but none produces a certificate of $\Gamma_3$ (\textit{e.g.}  $\lambdafour \ \oneach$ contains the 3-antichain as a convex-suborder). Therefore, $\Gamma_3$ has no certificates in $\Omega_{\mathbb{Z}}$ or in $\Omega_{\mathbb{Z}^-}$. Finally, note that $\Gamma_3$ does have a certificate in $\Omega_\mathbb{N}$, namely the order that contains the $\lambdacauset$ topped with an infinite chain. Therefore the infinite path containing $\Gamma_3$ only has certificates in $\Omega_\mathbb{N}$. Similarly, if there exists some $\Gamma_n\in \mathcal{P}$ all of whose infinite certificates are in $\Omega_{\mathbb{Z}^-}$ then $\mathcal{P}$ only has certificates in $\Omega_{\mathbb{Z}^-}$ (see for example the node $\{\ \vee \ , \threech\}$ and the infinite path that contains it). 
 
 The following proposition identifies the paths that have certificates in $\Omega_{\mathbb{Z}}$ and which are therefore of interest to us,
\begin{proposition}\label{prop240221} An infinite path $\mathcal{P}$ has a certificate in $\Omega_{\Z}$ if and only if every node in $\mathcal{P}$ has a certificate in $\Omega_{\Z}$. \end{proposition}

\begin{proof} Given an infinite path $\mathcal{P}=\Gamma_1\prec \Gamma_2\prec...$ each of whose nodes has a certificate in $\Omega_{\Z}$, the following inductive algorithm generates an infinite nested sequence of causal sets, $\lc{C}_{t_1}\subset \lc{C}_{t_2} \subset...$, whose ground-sets $[r_1,s_1],[r_2,s_2],...$ respectively, satisfy $r_1>r_2>....$ and $s_1<s_2<...$:
\newline \textit{Step 1:}
\newline 1.0) Pick some natural number $m_0>0$ and 
consider $\Gamma_{m_0}\in \mathcal{P}$.
\newline 1.1) By  lemma \ref{claim_theorem_1_prep}, there exists some $\Gamma_{m_1}\in \mathcal{P} $
that contains some certificate $C_{m_1}$ of $\Gamma_{m_0}$. Pick a representative $\lc{C}_{m_1}$ of $C_{m_1}$ and set $\lc{C}_{t_1}:= \lc{C}_{m_1}$.
\newline 1.2) Go to step 2.
\newline \textit{ Step} $k>1$:
\newline k.1) By lemma \ref{claim_theorem_1_prep}, there exists some $\Gamma_{m_{k}}\in\mathcal{P}$ that contains some certificate $C_{m_k}$ of $\Gamma_{t_{k-1}}\in\mathcal{P}$. Additionally, there exists a representative $\lc{C}_{m_k}$ of ${C}_{m_k}$ with ground-set $[p_k,q_k]$ that contains $\lc{C}_{t_{k-1}}$ as a sub-causet and satisfies at least one of $(a)$ $p_k<r_{k-1}$ or $(b)$ $q_k>s_{k-1}$. If there exists some $\lc{C}_{m_k}$ that satisfies both $(a)$ and $(b)$, set  $\lc{C}_{t_k}:= \lc{C}_{m_k}$. Otherwise, pick a representative $\lc{C}_{m_k}$ that satisfies $(a)$ or $(b)$. Go up one node along the path to $\Gamma_{1+m_k}\in\mathcal{P}$. Let $\lc{C}\in\lc{\Omega}_{\Z}$ be an infinite certificate of $\Gamma_{1+m_k}$ that contains $\lc{C}_{m_k}$ as a subcauset. Set $\lc{C}_{t_k}:=\lc{C}|_{[p_{k},q_{k}+1]}$ if $\lc{C}_{m_k}$ satisfies $(a)$ or  $\lc{C}_{t_k}:=\lc{C}|_{[p_{k}-1,q_{k}]}$  if $\lc{C}_{m_k}$ satisfies $(b)$.
\newline k.2) Go to step $k+1$.

By construction, the union $\lc{C}:=\bigcup_{i=1}^{\infty} \lc{C}_{t_i}\in\lc{\Omega}_{\Z}$ is a labeled certificate of $\mathcal{P}$. Therefore, if every node in $\mathcal{P}$ has a certificate in $\Omega_{\Z}$ then $\mathcal{P}$ has a certificate in $\Omega_{\Z}$. That the converse is true follows from definition \ref{def_240221}.
\end{proof}

Finally, we can define:
\begin{definition}
    $\Z$-covtree is the subtree of convex-covtree that contains exactly all nodes that have a certificate in $\Omega_{\Z}$.
\end{definition}

$\Z$-covtree is the two-way infinite analogue of covtree that we have set out to build. Proposition \ref{prop240221} guarantees that every inextendible path in $\Z$-covtree has at least one certificate in $\Omega_{\Z}$ and thus allows for every random walk on $\Z$-covtree to be interpreted as a dynamics with sample space  $\Omega_{\Z}$. To see the relationship between a walk on $\Z$-covtree and the corresponding dynamics, for each $\Gamma_n$ in $\Z$-covtree let $cert_{\Z}(\Gamma_n)\subset \lc{\Omega}_{\Z}$ denote the set of labeled  certificates of $\Gamma_n$ whose ground set is $\Z$. Let $\Sigma$ be the $\sigma$-algebra generated by all the $cert_{\Z}(\Gamma_n)$'s. A dynamics is then the probability measure space $({\Omega}_{\Z},\Sigma,\mathbb{P})$ where the measure $\mathbb{P}$ is given by $\mathbb{P}(cert_{\Z}(\Gamma_n))=\mathbb{P}(\Gamma_n)$. We will now show that the observables of these dynamics (\textit{i.e.} the elements of $\Sigma$) are the convex-events.

Recall that, for each finite order $C_n$, $convex(C_n)\subset\lc{\Omega}_{\mathbb{Z}}$ is the collection of causets that contain $C_n$ as a convex-suborder. Let $\mathcal{R}(\mathcal{C})$ denote the $\sigma$-algebra generated by the $convex(C_n)$'s. A convex-event is an element of $\mathcal{R}(\mathcal{C})$.\footnote{It may seem that  labeled causets have snuck back into the story. However, though in Section \ref{section_seq_growth} we formally defined a convex-event to be a set of labeled causets,  because the definition of $convex(C_n)$ is label independent, the  $convex(C_n)$'s and the convex-events generated by them are covariant and can be thought of -- in the obvious way -- as subsets of ${\Omega}_{\mathbb{Z}}$ -- i.e. sets of orders.}

\begin{lemma}\label{lemma_measure_space}
$\Sigma=\mathcal{R}(\mathcal{C})$.
\end{lemma}
\begin{proof}
\par We will show that any $convex(C_n)$ can be constructed by finite set operations on the $cert_{\Z}(\Gamma_m)$'s and vice versa, and the result follows.
\par Consider an $n$-order $B_n$. Let $\Gamma_n^i$ be the nodes in convex-covtree that contain $B_n$, where $i$ labels the individual nodes. Suppose $E\in cert_{\Z}(\Gamma_n^i)$ for some $i$. Then $B_n$ is an $n$-convex-suborder in $E$ and hence $E\in convex(B_n)$. Suppose $E\notin cert_{\Z}(\Gamma_n^i)$ for all $i$. Then $B_n$ is not an $n$-convex-suborder in $E$ and hence $E\notin convex(B_n)$. It follows that $convex(B_n)=\bigcup_i cert_{\Z}(\Gamma_n^i)$.
\par Consider some node $\Gamma_n=\{A_n^1, ..., A_n^k\}$ in convex-covtree. Let $\Omega(n)\setminus \Gamma_n =\{B_n^1, ..., B_n^l\}$. Suppose $E\in cert_{\Z}(\Gamma_n)$. Then $A_n^1,...,A_n^k$ are $n$-convex-suborders in $E$, and $B_n^1,...,B_n^l$ are not $n$-convex-suborders in $E$. Hence $E\in \bigcap \limits_{i=1}^k convex(A_n^i)\setminus \bigcup \limits_{j=1}^l convex(B_n^j)$. Suppose $E\notin cert_{\Z}(\Gamma_n)$. Then either $(i)$ there exists some $A_n^i\in \Gamma_n$ that is not an $n$-convex-suborder in $E \implies E\notin \bigcap \limits_{i=1}^k convex(A_n^i)$, or $(ii)$ there exists some $B_n^j\in \Omega(n)\setminus \Gamma_n$ that is an $n$-convex-suborder in $E \implies E\in \bigcup \limits_{j=1}^l convex(B_n^j)$. It follows that, $cert_{\Z}(\Gamma_n)=\bigcap \limits_{i=1}^k convex(A_n^i)\setminus \bigcup \limits_{j=1}^l convex(B_n^j)$.
\end{proof}

Lemma \ref{lemma_measure_space} strengthens the analogy between covtree and $\Z$-covtree---the observables of covtree are the stem-events while the observables of $\Z$-covtree are the convex-events. $\Z$-covtree is to alternating poscau what covtree is to labeled poscau. Convex-suborders are to two-way infinite dynamics what stems are to past-finite dynamics.
\FloatBarrier

\section{Discussion} \label{section_discussion}

In this work, we set out to build frameworks for growth dynamics for two-way infinite causal sets. We began by adapting the sequential growth paradigm to create \textit{alternating growth} models. We discussed the difficulties in attributing any physical significance to the \textit{process} of alternating growth and difficulties in  formulating and interpreting a  ``causality'' condition in this framework. We showed that the only alternating CSG model that satisfies Discrete General Covariance is Alternating Transitive Percolation. These may be considered as evidence against the existence of physically meaningful dynamical growth models for two-way infinite causal sets. 

On the positive side,  we identified a set of covariant observables that possess a clear physical interpretation, namely the convex-events. However we also showed that
that Alternating Transitive Percolation is deterministic with respect to the convex-events: the probability of any convex-event in Alternating Transitive Percolation is 0 or 1 and in particular the probability of any finite order being a convex-suborder of the growing causet is 1. There do exist Alternating CSG models for which this is not the case, suggesting that there may be models in which the convex-events may yet form a rich and interesting class of observables. This depends on future developments and whether some physically motivated and interesting alternating sequential growth models can be found. 

We then used the notion of convex-suborders and convex-events to adapt the covariant growth framework of \cite{Dowker:2019qiz} to two-way infinite growth. We encountered additional complications that are not present in the original construction, namely that the existence of a finite certificate does not guarantee the existence of an infinite certificate and that the existence of an infinite certificate does not guarantee the existence of a certificate in $\Omega_{\Z}$. Nevertheless, we were able to define a consistent covariant framework for two-way growth, $\Z$-covtree, whose observables are the convex-events.

Throughout, we were led to considering convex-suborders as the basic physical properties for two-way infinite growth by pursuing an analogy with stems and the role that they play in past-finite growth. In fact, convex-suborders are a generalisation of stems---a stem is a convex-suborder that contains its own past.\footnote{When considering both the $convex(C_n)$'s and the $stem(C_n)$'s as subsets of $\lc{\Omega}_{\mathbb{N}}$, a convex-event is a special case of a stem-event.}  Nevertheless, there may be other entities that could be considered as physical properties for two-way infinite dynamics, for example, downsets (subcausets that contain their own past---a generalisation of stem in which the condition of finite cardinality is relaxed), Moment of Time Surfaces (thickenned antichains \cite{Major:2005fy}), or intervals (special cases of convex-suborders). While these alternatives may prove fruitful in the future, we can identify a property unique to convex-suborders that is essential for our constructions: every infinite order contains at least one $n$-convex-suborder for every $n>0$. 

A significant downside of our new covariant framework is that the event that the completed order contains a post is not measurable since it is not a convex-event (Fig. \ref{no_posts}). Moreover, the cosmic renormalisation transformation associated with posts relies crucially on the cardinality of the past of the post,  while a post in a two-way infinite order will necessarily have an infinite past. Both posts and cosmic renormalisation play a pivotal role in the conception of causal set cosmology \cite{Sorkin:1998hi,Martin:2000js,Zalel:2020oyf} and so the two-way infinite growth models for causal set cosmology will require a new way of thinking about this cosmological paradigm.

Another challenge is to identify Alternating CSG dynamics in which there is a large and rich enough class of convex-events that serve usefully to discriminate between different realisations of the process, including with measures that lie strictly between 0 and 1.  To this end we may need to consider the sequence $(p_n)$, a representation of the CSG models that is related to the $t_k$'s by equation \eqref{eff_par_eq}. When $(p_n)$ is a constant sequence, the dynamics is Transitive Percolation and the measure of every convex-event is equal to 1. What behaviour does the sequence $(p_n)$ need to display in order for a dynamics to be probabilistic with respect to convex-events? How quickly must the sequence $(p_n)$ increase or decrease to give sufficiently different behaviour from the constant sequence of Transitive Percolation?

\begin{figure}[h]
    \centering
    \includegraphics[scale=0.6]{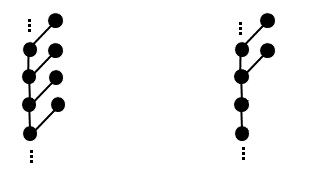}
    \caption{The 2-way infinite comb (left) and the future infinite comb above an infinite chain (right) have the same $n$-convex-suborders for every $n>0$, therefore every convex-event contains either both or neither. The order on the right contains posts while the order on the left does not. Therefore, the event that the completed order contains a post is not a convex-event.}
    \label{no_posts}
\end{figure}

\vspace{2mm}
\par \textbf{Acknowledgments:} The authors thank Benjamin Honan and Ashim Sen Gupta for discussions. FD thanks the participants of the online seminar series ``Cosmology Beyond Spacetime''  on 10 March 2021, organised by Nick Huggett and Christian W\"uthrich for discussions.  FD acknowledges the support of  the Leverhulme/Royal Society interdisciplinary  APEX grant APX/R1/180098. FD
is supported in part by Perimeter Institute for Theoretical Physics. Research at Perimeter Institute is supported by the Government of Canada through
Industry Canada and by the Province of Ontario through the Ministry of Economic Development and Innovation. FD is supported in part by STFC grant
ST/P000762/1.  SZ thanks the Perimeter Institute for hospitality while this work was being completed. SZ is partially supported by the Beit Fellowship for Scientific Research and by the Kenneth Lindsay Scholarship Trust.

\appendix
\section{Table of symbols defined in the text}
\FloatBarrier
\begin{table}[htpb!]\centering
\begin{tabular}{| l | l |}
\hline
  $\lc{C}, \lc{D}, ...$ &  labeled causets \\ \hline
  $C, D,...$ &  orders \\ \hline
  $\cong$ &  $\lc{C}\cong\lc{D}$ if $\lc{C}$ and $\lc{D}$ are equal up to an order-isomorphism\\ \hline
  $\tilde{\Omega}_{\mathbb{N}}$ & The set of labeled causets with ground-set $\mathbb{N}$\\ \hline
  $\tilde{\Omega}_{\mathbb{Z}}$ & The set of labeled causets with ground-set $\mathbb{Z}$\\ \hline
  $\tilde{\Omega}_{\mathbb{Z}^-}$ & The set of labeled causets with ground-set $\mathbb{Z}^-$\\ \hline
  $\tilde{\Omega}$ & The set of infinite labeled causets, $\tilde{\Omega}\equiv\tilde{\Omega}_{\mathbb{N}}\sqcup\tilde{\Omega}_{\mathbb{Z}}\sqcup\tilde{\Omega}_{\mathbb{Z}^-}$\\ \hline
  ${\Omega}$ & The set of infinite orders, ${\Omega}:=\tilde{\Omega}/\cong$\\ \hline
  ${\Omega}_{\mathbb{N}}$ &  The set of orders that have a representative in $\lc{{\Omega}}_{\mathbb{N}}$ \\ \hline
  ${\Omega}_{\mathbb{Z}}$ &  The set of orders that have a representative in $\lc{{\Omega}}_{\mathbb{Z}}$\\ \hline
  ${\Omega}_{\mathbb{Z}^-}$ & The set of orders that have a representative in $\lc{{\Omega}}_{\mathbb{Z}^-}$\\ \hline

  $\Omega(n)$ & The set of $n$-orders for some $n\in\mathbb{N}^+$ \\ \hline
  $\Gamma_n$ &  A subset of $\Omega(n)$\\ \hline
\end{tabular}
\caption{Table of symbols defined in the text.}
\end{table}
\FloatBarrier

\section{On infinite certificates of nodes and paths in convex-covtree}

By proposition \ref{lemma_1_211020}, every infinite path in convex-covtree has at least one certificate in $\Omega$. By proposition \ref{prop240221}, an infinite path in convex-covtree has a certificate in $\Omega_{\mathbb{Z}}$ if and only if each of its nodes has a certificate in $\Omega_{\mathbb{Z}}$. There exist nodes whose infinite certificates are only contained in $\Omega_{\mathbb{N}}$ or only in $\Omega_{\mathbb{Z}^-}$ (see section \ref{subsec_inf_paths} for examples), and therefore the infinite paths containing these nodes only have certificates in   $\Omega_{\mathbb{N}}$ or in $\Omega_{\mathbb{Z}^-}$, respectively.

There exists no node in convex-covtree whose infinite certificates are only contained in $\Omega_{\mathbb{Z}}$, since if a node has a certificate in $\Omega_{\mathbb{Z}}$ then it has a certificate in $\Omega_{\mathbb{N}}$ and in $\Omega_{\mathbb{Z}^-}$. To see this, let $\lc{C}\in\lc{\Omega}_{\mathbb{Z}}$ be a labeled certificate of some $\Gamma_n$ and let $\lc{C}|_{[k,l]}$ be a finite certificate of $\Gamma_n$. Then $\lc{C}|_{[k,\infty)}$ is order-isomorphic to some $\lc{D}\in\lc{\Omega}_{\mathbb{N}}$ and $\lc{D}$ is a certificate of $\Gamma_n$. Similarly, $\lc{C}|_{(\infty, l]}$ is order-isomorphic to some $\lc{E}\in\lc{\Omega}_{\mathbb{Z}^-}$ and $\lc{E}$ is a certificate of $\Gamma_n$.


There exist infinite paths in convex-covtree whose infinite certificates are only contained in $\Omega_{\mathbb{Z}}$. An infinite path only has certificates in $\Omega_\mathbb{Z}$ if and only if there is no one order in $\Omega_\mathbb{N}\cup\Omega_{\mathbb{Z}^-}$ that is a certificate of \textit{every} node in the path. For example, consider the path 
\vspace{1mm}
$$\mathcal{P}=\{\ \oneach \}\prec \{\twoch, \ \twoach\}\prec\{\threech, \lambdacauset, \vee  \}\prec\{ \fourch, \lambdafour,\topvee , \diamond \ \}\prec ... \ ,$$
 whose certificate is the order $D$ shown on the right of figure \ref{path_cert_fig}. Each node in $\mathcal{P}$ has a certificate in  $\Omega_{\mathbb{N}}$, as illustrated in figure \ref{path_cert_fig}, but there is no order in $\Omega_{\mathbb{N}}$ that is a certificate of every node in  $\mathcal{P}$. One way to see this is to notice that for every $n>3$, $\Gamma_n\in\mathcal{P}$ has a unique minimal certificate, namely the diamond sandwiched between two $(n-3)$-chains. Now, pick some $n>3$ and w.l.g. pick a representative of its minimal certificate, $\lc{C}_{2n-6}$, with ground-set $[0,2n-6]$. We seek a labeled minimal certificate $\lc{C}_{2n-4}$ of $\Gamma_{n+1}$ that contains $\lc{C}_{2n-6}$ as a subcauset, and find that $\lc{C}_{2n-4}$ must have ground-set $[-1,2n-5]$. Next we seek a labeled minimal certificate $\lc{C}_{2n-2}$ of $\Gamma_{n+2}$ that contains $\lc{C}_{2n-4}$ as a subcauset, and find that $\lc{C}_{2n-2}$ must have ground-set $[-2,2n-4]$ \textit{etc.} Since at each stage we add a positive and a negative integer to the ground-set, in the infinite limit the labeled certificate must have ground-set $\mathbb{Z}$.
 
Since the existence of a certificate in $\Omega_{\mathbb{N}}$ for each $\Gamma_n\in\mathcal{P}$ does not guarantee that $\mathcal{P}$ has a certificate in $\Omega_{\mathbb{N}}$ (\textit{i.e.} there is no analogue of proposition \ref{prop240221} for $\Omega_{\mathbb{N}}$) there is no subtree of convex-covtree that contains exactly all infinite paths that have certificates in $\Omega_{\mathbb{N}}$ , \textit{i.e.} there is no $\mathbb{N}$ analogue of $\mathbb{Z}$-covtree. Thus, convex-covtree cannot be truncated into a growth framework whose sample space is $\Omega_{\mathbb{N}}$, suggesting that convex-events (now treated as subsets of $\Omega_{\mathbb{N}}$) are not rich enough to exhaust the set of observables in past-finite dynamics.

One can understand this difference between $\Omega_{\mathbb{N}}$ and $\Omega_{\mathbb{Z}}$ using metric space techniques. For any two orders $C$ and $D$, let $C\sim D$ if and only if $C$ and $D$ are a convex-rogue pair, \textit{i.e.} if they share the same $n$-convex-suborders for all $n$. Let $\Omega_{\mathbb{N}}/\sim$ and $\Omega_{\mathbb{Z}}/\sim$ be quotient spaces under the convex-rogue equivalence relation, so that their elements are equivalence classes of orders denoted by $[C]$ \textit{etc.} We can consider these quotient spaces as metric spaces with metric $d([C],[D])=\frac{1}{2^n}$, where $n$ is the largest integer for which representatives of $[C]$ and $[D]$ have the same sets of $n$-convex-suborders. Given a node $\Gamma_n$ in convex-covtree we can associate with it a subset $[cert_{\mathbb{N}}(\Gamma_n)]\subseteq\Omega_{\mathbb{N}}/\sim$, namely the set of elements of $\Omega_{\mathbb{N}}/\sim$ whose representatives are certificates of $\Gamma_n$, and similiarly $[cert_{\mathbb{Z}}(\Gamma_n)]\subseteq\Omega_{\mathbb{Z}}/\sim$. Given a path $\mathcal{P}=\Gamma_1\prec\Gamma_2\prec...$, we can associate with it the sets $[cert_{\mathbb{N}}(\mathcal{P})]=\bigcap_{\Gamma_n\in\mathcal{P}}[cert_{\mathbb{N}}(\Gamma_n)]$ and $[cert_{\mathbb{Z}}(\mathcal{P})]=\bigcap_{\Gamma_n\in\mathcal{P}}[cert_{\mathbb{Z}}(\Gamma_n)]$. Since the metric space $(\Omega_{\mathbb{Z}}/\sim,d)$ is complete, by Cantor's lemma $[cert_{\mathbb{Z}}(\mathcal{P})]$ is non-empty whenever all the $[cert_{\mathbb{Z}}(\Gamma_n)]$'s are non-empty (cf. proposition  \ref{prop240221}). On the other hand, the metric space $(\Omega_{\mathbb{N}}/\sim,d)$ is not complete and therefore $[cert_{\mathbb{N}}(\mathcal{P})]$ can be empty when all the $[cert_{\mathbb{N}}(\Gamma_n)]$'s are non-empty.

\begin{figure}[h]
    \centering
    \includegraphics[scale=0.6]{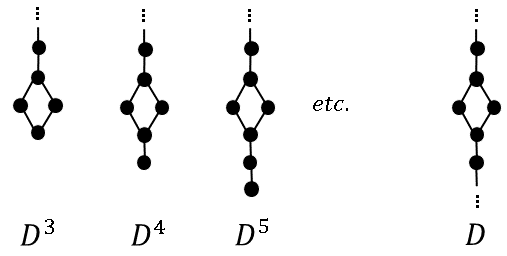}
    \caption{The order $D\in\Omega_{\Z}$ shown on the right is a certificate of the path $\mathcal{P}$. Every node in $\mathcal{P}$ has a certificate in $\Omega_{\mathbb{N}}$: $D^3$ is a certificate of $\Gamma_n\in\mathcal{P}$ only for $n\leq 3$, $D^4$ is a certificate of $\Gamma_n\in\mathcal{P}$ only for $n\leq 4$, $D^5$ is a certificate of $\Gamma_n\in\mathcal{P}$  only for $n\leq 5$, \textit{etc}. There is no order in $\Omega_{\mathbb{N}}$ that is a certificate of every node in $\mathcal{P}$.}
    \label{path_cert_fig}
\end{figure}

\bibliography{../Bibliography/refs}{}
\bibliographystyle{unsrt}

\end{document}